\documentclass[11pt]{article}

\usepackage{graphicx, mathtools, amssymb, latexsym, amsmath, amsfonts, amsthm, bbm, tikz}
\usepackage[toc,page]{appendix}

\usepackage{array, fancyhdr, bm, listings, soul, xcolor,titlesec, cancel,ifthen}

\usepackage{hyperref}
\hypersetup{
    colorlinks=true, %
    linkcolor=blue, %
    urlcolor=red, %
    linktoc=all %
}

\theoremstyle{plain}
\newtheorem{theorem}{Theorem}[section]
\newtheorem{lemma}[theorem]{Lemma}
\newtheorem{claim}[theorem]{Claim}
\newtheorem{proposition}[theorem]{Proposition}

\newtheorem{corollary}[theorem]{Corollary}

\newtheorem*{lemma*}{Lemma}
\newtheorem*{claim*}{Claim}
\newtheorem*{proposition*}{Proposition}
\newtheorem*{fact*}{Fact}
\newtheorem*{corollary*}{Corollary}
\newtheorem*{hint*}{Hint}

\theoremstyle{definition}
\newtheorem{definition}[theorem]{Definition}
\newtheorem{remark}[theorem]{Remark}

\newtheorem{question}[theorem]{Question}

\newtheorem*{theorem*}{Theorem}
\newtheorem*{definition*}{Definition}
\newtheorem*{remark*}{Remark}
\newtheorem*{notation*}{Notation}
\newtheorem*{example*}{Example}
\newtheorem*{examples*}{Examples}
\newtheorem*{question*}{Question}
\newtheorem*{answer*}{Answer}
\newtheorem*{problem*}{Problem}
\newtheorem*{solution*}{Solution}
\newtheorem*{idea*}{Idea}
\newtheorem*{conjecture*}{Conjecture}

\newcommand{\floor}[1]{\lfloor {#1} \rfloor}

\newcommand{\ceil}[1]{\lceil {#1}\rceil}

\newcommand\half{\frac12}

\newcommand{\ind}[1]{^{(#1)}}
\newcommand\defeq{\ensuremath{\stackrel{\rm def}{=}}} %

\renewcommand{\Pr}{\mathop{\bf Pr\/}}

\newcommand{\E}{\mathop{\bf E\/}}

\DeclareMathOperator*\poly{poly}
\newcommand\Enc{\textnormal{Enc}}
\newcommand\Dec{\textnormal{Dec}}

\newcommand\BDC{\textnormal{BDC}}

\newcommand{\eps}{\varepsilon}

\newcommand\rand{\ind{\textnormal{avg}}}

\newcommand{\avgtraces}{{\exp(O(\log^{1/3} n))}}

\newcommand{\ourtraces}{{\exp(O_q(\log^{1/3}(\frac{1}{\varepsilon})))}}

\mathtoolsset{showonlyrefs}

\begin{document}

\begin{titlepage}
\title{Coded trace reconstruction in a constant number of traces}
\author{Joshua Brakensiek\thanks{Department of Computer Science, Stanford University, Stanford, CA. Email: {\tt jbrakens@cs.stanford.edu}. Research supported by an NSF Graduate Research Fellowship.} , Ray Li\thanks{Department of Computer Science, Stanford University.  Research supported by an NSF Graduate Research Fellowship grant DGE-1656518 and NSF grant CCF-1814629. Email: {\tt rayyli@cs.stanford.edu}} , Bruce Spang\thanks{Department of Computer Science, Stanford University, Stanford, CA. Email: {\tt bspang@cs.stanford.edu}.}}
\date{\today}
\maketitle\thispagestyle{empty}

\begin{abstract}
The \emph{coded trace reconstruction} problem asks to construct a code $C\subset \{0,1\}^n$ such that any $x\in C$ is recoverable from independent outputs (``traces'') of $x$ from a binary deletion channel (BDC). We present binary codes of rate $1-\varepsilon$ that are efficiently recoverable from ${\exp(O_q(\log^{1/3}(\frac{1}{\varepsilon})))}$ (a constant independent of $n$) traces of a $\operatorname{BDC}_q$ for any constant deletion probability $q\in(0,1)$. We also show that, for rate $1-\varepsilon$ binary codes, $\tilde \Omega(\log^{5/2}(1/\varepsilon))$ traces are required. The results follow from a pair of black-box reductions that show that average-case trace reconstruction is essentially equivalent to coded trace reconstruction. We also show that there exist codes of rate $1-\varepsilon$ over an $O_{\varepsilon}(1)$-sized alphabet that are recoverable from $O(\log(1/\varepsilon))$ traces, and that this is tight.
\end{abstract}

\end{titlepage}

\section{Introduction}
The \emph{trace reconstruction} problem was first proposed in \cite{Levenshtein:2001fk, Levenshtein:2001vu} and further developed in \cite{Batu:2004um}. 
In trace reconstruction, we wish to recover an unknown binary string $x \in \{0, 1\}^n$ given a few random subsequences of $x$.
Each subsequence, or \emph{trace}, is generated by sending $x$ through the \emph{binary deletion channel with deletion probability $q$} ($\BDC_q$), which independently deletes each symbol of $x$ with probability $q \in (0, 1)$. In particular, the positions of the deleted bits are not known. For example, deleting either the first or second bit of ``110'' gives the trace ``10''.

Trace reconstruction has been primarily studied in two settings: \emph{worst-case}, in which the input string $x$ is chosen adversarially, and \emph{average-case}, when the input string $x$ is chosen uniformly at random over all possible $n$-bit strings. The fundamental question in both settings is to determine the minimum number of traces $T=T(n)$ needed in order to recover a length $n$ string $x$ correctly with high probability. In both settings, there is currently an exponential gap (as a function of $n$) for bounding $T(n)$ -- see Section~\ref{subsec:rel-work} for the best known bounds.

In this work, we consider an emerging \cite{Haeupler:2014fn, Cheraghchi:2019vd, Abroshan:2019wt} variant of the trace reconstruction known as \emph{coded trace reconstruction}. In this model, we want the smallest $T$ such that there exists a high rate code $C\subset \{0,1\}^n$ such that, for an adversarially chosen $x\in C$, we can recover $x$ with high probability from $T$ traces. This model is directly motivated by DNA storage~\cite{yazdi2017, Cheraghchi:2019vd}, in which data is stored as multiple encoded strands of DNA. Besides directly generalizing the trace reconstruction problem, coded trace reconstruction also generalizes the well-studied problem of determining the capacity of the binary deletion channel.

In this coded setting, we wish to design codes for trace reconstruction with high \emph{rate}, which is defined\footnote{All $\log$s and $\exp$s are base $2$ unless otherwise specified.} to be $\log|C|/n$.
We consider the regime in which the rate is $1 - \eps$ (i.e., $|C| \approx 2^{(1-\eps)n}$), where $\eps \in (0, 1)$ is a small constant or shrinking as a function of $n$. In particular, the key question we study is as follows.

\begin{question}
    For a given $\eps \in (0, 1)$ and positive integer $n$, what is the smallest $T$ such that we can construct a binary code of rate $1-\eps$ and length $n$ recoverable from $T$ traces?
\end{question}

\paragraph*{Contributions.}
We summarize the main contributions of our work below. See Section~\ref{sec:contributions} for formal theorem statements. In all these results, we consider any constant $q \in (0, 1)$. 
\begin{enumerate}
\item \textbf{Binary codes with constant number of traces.} For $\eps \in (0, 1)$, we construct an infinite family of binary codes of rate $1-\varepsilon$ efficiently recoverable from a constant number of traces over the $\BDC_q$ (independent of $n$). This follows as an immediate corollary (Corollary~\ref{cor:main}) of the following more general result we prove.

\item \textbf{Black-box upper bounds from average-case trace reconstruction.} We show that, if average-case trace reconstruction on length $n$ strings succeeds with sufficiently high probability in $T(n)$ traces, then there exist rate $1-\varepsilon$ codes that are decodable from $T(\tilde O_q(1/\varepsilon))$ traces over the $\BDC_q$ (Theorem~\ref{thm:main}). 
  In particular, by a result in \cite{HPP18}, $\ourtraces < \frac{1}{\varepsilon^{o(1)}}$ traces suffice (Corollary~\ref{cor:main}). 
  
  \item \textbf{Black-box lower bounds from average-case trace reconstruction.} Conversely, we show that if average-case reconstruction on length $n$ strings requires $T(n)$ traces, then reconstruction of any binary code of rate $1-\eps$ requires $T(\tilde \Omega_q(1/\sqrt{\varepsilon}))$ traces over the $\BDC_q$ (Theorem~\ref{thm:lb}).
  In particular, by a recent result \cite{Chase:2019tb}, $\tilde \Omega_q(\log^{5/2}(1/\eps))$ traces are required (Corollary~\ref{cor:lb}).
  
  \item \textbf{Near-equivalence of average-case and coded trace reconstruction.} The two black-box reductions together imply that estimating the optimal number of traces for a code of rate $1-\eps$ is equivalent to closing the lower and upper bounds within a polynomial for average-case trace reconstruction on strings of length $\poly(1/\eps)$ (Remark~\ref{rem:code-avg}). 

  \item \textbf{Optimal number of traces for constant-sized alphabet.} We also consider the coded trace reconstruction problem over larger alphabets than binary. In particular, we give rate\footnote{The rate of a code $|C|$ of length $n$ over an alphabet $\Sigma$ is $\frac{\log_{|\Sigma|}|C|}{n}$} $1-\varepsilon$ codes over an alphabet of size $O_\varepsilon(1)$ that are efficiently encodable and decodable from $O(\log_{1/q}(1/\varepsilon))$ traces (Theorem~\ref{thm:large-ub}). We show this is optimal up to a constant factor (Theorem~\ref{thm:large-lb}). This shows that coded trace reconstruction is strictly easier for larger alphabets than for binary alphabets. To the best of our knowledge, this is the first non-trivial tight result in \em any \em model of trace reconstruction for the deletion channel.
\end{enumerate}

\subsection{Related work}\label{subsec:rel-work}

We now discuss how our results are situated at the intersection of the trace reconstruction and coding theory literature.

\paragraph{Classical trace reconstruction.}

One of the main motivations for trace reconstruction is the application to DNA sequencing in computational biology \cite{Batu:2004um}. When DNA is sequenced, the results may have insertion, deletion, and substitution errors. The original goal of trace reconstruction was to understand a simplified model of how an unknown piece of DNA can be recovered from its sequences. Recently, sequencing has been used for DNA storage \cite{yazdi2017, Cheraghchi:2019vd}, in which data is encoded so that it can be stored in DNA. This code needs to be decodable using a trace reconstruction-like process, while being high rate and using as few traces as possible.

The theoretical worst-case setting of trace reconstruction, recovering an arbitrary binary string, was originally studied in \cite{Levenshtein:2001fk, Levenshtein:2001vu, Batu:2004um,Holenstein:2008wy}. 
The current state of the art was derived independently in \cite{De:2017jj} and \cite{Nazarov:2017vz}, who show that $\exp(O(n^{1/3}))$ traces suffice for any constant deletion probability $q\in(0,1)$. A very recent result~\cite{chase2020upper} shows that $\exp(O(n^{1/5}))$ traces suffice for any $q \in (0,1/2]$.
Several works have also considered lower bounds for worst-case trace reconstruction \cite{Batu:2004um,Holenstein:2008wy,MPV14,Holden:2018tx,Chase:2019tb}. The best known lower bound is $\Omega\left(\frac{n^{3/2}}{\log^{16} n}\right)$ traces \cite{Chase:2019tb}, which has an exponential gap compared to the best known upper bound. Our work does not use or address worst-case trace reconstruction.

In the average-case setting studied by \cite{Holenstein:2008wy,MPV14,PeresZ17, HPP18}, the best upper bound is given by \cite{HPP18}, who showed that, for all deletion probabilities $q\in(0,1)$, a subpolynomial $\avgtraces$ traces suffice to recover a random string with high probability. 
Several works have also considered lower bounds for average-case trace reconstruction \cite{MPV14,Holden:2018tx,Chase:2019tb}.
The current best bound of $\Omega\left(\frac{\log n^{5/2}}{(\log\log n)^{16}}\right)$ traces \cite{Chase:2019tb} again has an exponential gap.
Our work shows that resolving the optimal number of traces up to a constant factor for coded trace reconstruction is essentially equivalent to average-case reconstruction. 

Trace reconstruction over a larger alphabet is less well studied. \cite{MPV:2014el,De:2017jj} show that it is possible to turn any trace reconstruction algorithm over a non-binary alphabet into a trace over a binary alphabet and use binary trace reconstruction to solve the problem, at a small cost to the failure probability. 
For coded trace reconstruction, we show that there is a substantial benefit to using a non-binary alphabet. For constant-sized alphabets, we show a matching upper and lower bound, determining the optimal number of traces up to a constant factor.

\paragraph{Coded trace reconstruction.}
Coded trace reconstruction generalizes the classical questions above about trace reconstruction.
The worst-case trace reconstruction question over a binary alphabet asks how many traces $T(n)$ are needed to achieve error probability $o(1)$ for the code $C=\{0,1\}^n$. As we show in Section~\ref{sec:avg-to-coded}, average-case trace reconstruction is equivalent to asking how many traces $T(n)$ are needed to achieve error probability $o(1)$ for a code $C$ of size $2^n(1-o(1))$. We use this connection to average-case trace reconstruction to construct much longer codes which are recoverable from few traces.

Cheraghchi, Gabrys, Milenkovic, and Ribeiro \cite{Cheraghchi:2019vd} formulated the coded trace reconstruction problem considered here.
Among other constructions, they give explicit constructions of binary codes of rate $1-O(\frac{1}{\log \log n})$ recoverable in $\exp(O(\log\log n)^{2/3})$ traces, and rate $1-O(\frac{1}{\log n})$ code recoverable in $\poly\log n$ traces. Our work improves the number of traces and allows a wider range of rates. For any $\varepsilon \ge n^{-o(1)}$, we show that there exist binary codes of rate $1-\varepsilon$ recoverable in $\ourtraces$ traces. Taking $\varepsilon = \Theta(\frac{1}{\log\log n})$ and $\varepsilon = \Theta(\frac{1}{\log n})$ gives the respective improvements to \cite{Cheraghchi:2019vd} in the number of traces. We emphasize that all the constructions of \cite{Cheraghchi:2019vd} have polynomial time encoding and decoding, whereas our constructions have polynomial time decoding in all considered parameter settings, but only polynomial time encoding when $\varepsilon\ge \Omega(\frac{\log \log n}{\log n})$.

Although our work deals with a constant fraction of deletions, several prior works considered coding for trace reconstruction for small numbers of deletions. Haeupler and Mitzenmacher \cite{Haeupler:2014fn} showed that, for any fixed integer $T$, as the deletion probability $q$ approaches 0, there exists a binary code recoverable from $T$ traces across the $\BDC_q$ with rate $1-O(H(q^T))$, where $H$ is the binary entropy function.
By contrast, our codes handle deletion probabilities arbitrarily close to 1.
We show, for example, that there exist binary codes of rate $0.99$ recoverable from $T=O(1)$ traces of the $\BDC_{0.99}$.
Abroshan, Venkataramanan, Dolecek, and Guill\'en \cite{Abroshan:2019wt} consider coding for channels applying a constant number of deletions.
They concatenate $\ell$ Varshamov-Tenengolts~\cite{varshamov1965codes} codes of length $m$ to construct a code of length $m\ell$ and rate $1 - O(\frac{\log m}{m})$ for any $m,\ell\ge 1$.
They bound the error probability for recovering for a channel that applies exactly $\ell'$ deletions, when $\ell' < \ell$.

\paragraph{Other trace reconstruction variants.}

There has recently been a variety of work on other problems related to trace reconstruction, which our work does not address. \cite{Gabrys:2018tj} considers the problem of recovering a string from the multiset of all its length $L$ substrings. \cite{BanCFSS19} studies population recovery under the deletion channel, an extension to trace reconstruction where we recover an unknown distribution over input strings, rather than a single input string. In \cite{KrishnamurthyMMP19}, the authors consider the problems of reconstructing matrices and sparse strings from traces.

\paragraph{Codes for the deletion channel.}
The optimal rate for coded trace reconstruction with one trace is also known as the \emph{capacity} of the binary deletion channel, a well-studied and difficult problem. 
The capacity of the binary deletion channel with deletion probability $q$ is clearly at most $1-q$, the capacity of the simpler binary erasure channel.
When $q\to 0$, the capacity is known to approach $1-H(q)$, where $H(q)$ is the binary entropy function (see \cite{DG2001} for the lower bound and \cite{KanoriaM2013, KMS2010} for the upper bound).
When $q\to 1$, the capacity is known to be $\Theta(1-q)$, but the exact capacity is known only to be roughly between $0.11(1-q)$ \cite{DrineaM06, DrineaM07}, and $0.41(1-q)$ \cite{RahmatiD15}.
A polynomial time encodable/decodable code meeting this up to a constant factor was given in \cite{GuruswamiL19, ConS19}.
The current best capacity upper bounds for intermediate $q$ (e.g., $q=0.5$) are given by \cite{FertonaniD10, RahmatiD15, Cheraghchi18}. We incorporate techniques used in constructing codes for the binary deletion channel in our construction of Theorem~\ref{thm:main}.
Our work shows that, at $q=1-\delta$, if one is allowed to reconstruct from $O_\delta(1)$ traces of the BDC$_q$ rather than only one trace, the capacity of the resulting channel improves from $\Theta(\delta)$ to $0.99$.

\subsection{Main results}\label{sec:contributions}

We now define the coded trace reconstruction problem formally and state our main theorems. 
For $q\in(0,1)$ and $x\in\{0,1\}^n$, we let $\BDC_q(x)$ denote the probability distribution of output of $x$ across the $\BDC_q$.
We let $\{0,1\}^*$ denote the set of binary strings of any length.

  \begin{definition}
  For $q,\delta\in(0,1)$ and positive integers $n$ and $T$, we say a code $C\subset \{0,1\}^n$ is \emph{$(T,q,\delta)$  trace reconstructible} if there exists a \emph{decoding function} $\Dec:(\{0,1\}^*)^T\to C$ such that, for all $c\in C$,
  \begin{align}
    \Pr_{z_1,\dots,z_T\sim \BDC_q(c)}[\Dec(z_1,\dots,z_T) \neq c] < \delta.
  \end{align}
  \end{definition}
  Typically, we desire $\delta\to 0$ as $n\to\infty$. 
  We say $C$ is \emph{decodable} in time $t$ if $\Dec$ can be computed in time $t$.
  We say $C$ is \emph{encodable} in time $t$ if there exists a bijection $\Enc:\{1,\dots,|C|\}\to C$ that can be evaluated in time $t$.
  The following notation, denoting the optimal number of traces for average-case trace reconstruction, is used throughout the paper.
  \begin{definition}\label{def:rand-tr}
    For $m\ge 1, q\in(0,1),$ and $\beta\ge 0$, let $T\rand_{q,\beta}(m)$ denote the smallest integer $T$ such that there exists a trace reconstruction algorithm for the $\BDC_q$ using $T$ traces that, on a uniformly random string $x$ of length $m$, succeeds with probability (over the randomness of the string and channel) at least $1-\frac{1}{3m^{\beta}}$.
    When $\beta$ is omitted, we take $\beta=0$.
  \end{definition}
  By repetition of the reconstruction algorithm and subsequently taking a majority vote, we have $T\rand_q(m)\le T\rand_{q,\beta}(m)\le O(\beta\log m)\cdot T\rand_q(m)$, so $T\rand_{q,\beta}(m)$ and $T\rand_q(m)$ are roughly the same size for constant $\beta$.

  \paragraph{Binary upper bound.}
  We prove the following upper bound for coded trace reconstruction, which allows bounds for average-case trace reconstruction to be turned into bounds for coded trace reconstruction.
  \begin{theorem}
    \label{thm:main}
    For all $q,\varepsilon\in(0,1)$, there exists constants $n_0=1/\varepsilon^{O_q(1)}$, $\beta=\Theta_q(1), n_{R}=\Theta_q(\frac{1}{\varepsilon}\log\frac{1}{\varepsilon})$, and $\delta=2^{-\varepsilon^{O_q(1)}n}$ such that, for all $n\ge n_0$, there exists a code $C\subset \{0,1\}^n$ of rate $1-\varepsilon$ that is $(T\rand_{q,\beta}(n_R),q,\delta)$ trace reconstructible.
    Furthermore, the encoding can be done in time $\poly_{\varepsilon,q}(n)$ and trace reconstruction can be done in time $\poly(n)$.
  \end{theorem}
  We can instantiate Theorem~\ref{thm:main} using the state-of-the-art construction for average-case trace reconstruction of Holden, Pemantle, and Peres \cite{HPP18}, which states that $T\rand_q(\frac{1}{\varepsilon})\le \exp(O_q(\log^{1/3}\frac{1}{\varepsilon}))$. Doing so gives the following.
  \begin{corollary}
    \label{cor:main}
    For all $q,\varepsilon \in (0,1)$, there exists constants $n_0=1/\varepsilon^{O_q(1)}, T=\ourtraces,$ and $\delta=2^{-\varepsilon^{O_q(1)}n}$ such that, for all $n\ge n_0$, there exist codes of length $n$ and rate at least $1-\varepsilon$ that are $(T,q,\delta)$ trace reconstructible.
  \end{corollary}
  \begin{remark}
    \label{rem:ub-eps}
    In coding theory, we are sometimes interested in codes with rate quickly approaching 1, and our bounds on the number of traces hold in this setting as well.
    For every $q \in (0, 1)$, Theorem~\ref{thm:main} and Corollary~\ref{cor:main} holds for all integers $n\ge \frac{1}{\varepsilon^{\Omega_q(1)}}$.
    Thus, we obtain obtain similar results for $\varepsilon$ going to 0 with $n$ so long as $\varepsilon\ge \frac{1}{n^{O_q(1)}}$. 
    Setting $\varepsilon = O(\frac{1}{\log n})$, we have codes of rate $1 - O(\frac{1}{\log n})$ recoverable from $\exp(O_q(\log\log n)^{1/3})$ traces with failure probability $2^{-\tilde O_q(n)}$, improving upon the $\poly\log n$ number of traces in \cite{Cheraghchi:2019vd} needed for the same $\varepsilon$.
    Our construction also gives a better bound on the number of traces when  $\varepsilon=O(\frac{1}{\log\log n})$, improving from $\exp(O_q(\log\log n)^{2/3})$ traces to $\exp(O_q(\log\log\log n)^{1/3})$ traces. 
  \end{remark}
  \begin{remark}
    While we improve on the number of traces in \cite{Cheraghchi:2019vd} and also give polynomial time decoding like in \cite{Cheraghchi:2019vd}, their codes are all polynomial time encodable, whereas ours are only so when $\varepsilon \ge \Omega(\frac{\log\log n}{\log n})$: a careful look at our runtimes shows our code is encodable in time $t_{enc}(\Theta(\frac{1}{\varepsilon}\log\frac{1}{\varepsilon}))\cdot \poly n$, where $t_{enc}(n')$ is the amount of time needed to encode a string of length $n'$ used for average-case trace reconstruction, as in Lemma~\ref{lem:short}. 
    Naively we upper bound $t_{enc}(n')\le 2^{O(n')}$.
    Thus, when $\varepsilon=O(\frac{1}{\log n})$, while we improve on the number of traces from \cite{Cheraghchi:2019vd} and also give polynomial time decoding, only \cite{Cheraghchi:2019vd} has codes with both encoding and reconstruction in polynomial time.
    Furthermore, the constants in our code are quite large, making them currently impractical. 
    Still, we hope the ideas in our construction could be used for future efficient constructions.
  \end{remark}

  \paragraph{Binary lower bound.}
  We also prove the following converse, showing that the number of traces needed for rate $1-\varepsilon$ trace reconstruction is at least the number of traces needed for average-case trace reconstruction on length $\frac{1}{\varepsilon^{1/2-o(1)}}$ strings with failure probability 1/3. 
  \begin{theorem}
  \label{thm:lb}
    For all $q,\delta\in(0,1)$, for sufficiently small $\varepsilon>0$, there exists $m=\tilde\Omega_q(\frac{1}{\varepsilon^{1/2}})$ such that, if $T=T_{q}\rand(m)$, all rate $1-\varepsilon$ codes of sufficiently large length are \emph{not} $(T-1,q,\delta)$-trace reconstructible.
  \end{theorem}
  Using Theorem~\ref{thm:lb}, we can adapt the state-of-the-art lower bound for average case trace reconstruction into a lower bound for coded trace reconstruction. Recently Chase~\cite{Chase:2019tb}, building off work of Holden and Lyons \cite{Holden:2018tx}, showed that $T_{q}\rand(m)\ge \tilde\Omega_q((\log m)^{5/2})$.\footnote{Here, $\tilde\Omega(\cdot)$ suppresses $\log\log$ factors. In fact, they show something stronger: even achieving success probability $\exp(m^{-0.15})$ requires that many traces.}
  Applying Theorem~\ref{thm:lb} to this result gives us the following lower bound.
  \begin{corollary}
    \label{cor:lb}
    For all $q,\delta\in(0,1)$ and $\varepsilon>0$ sufficiently small, there exists $T = \tilde\Omega_q((\log \frac{1}{\varepsilon})^{5/2})$ such that all rate $1-\varepsilon$ codes of sufficiently large length are \emph{not} $(T,q,\delta)$-trace reconstructible.
  \end{corollary}
  \begin{remark}
    \label{rem:lb-eps}
    Theorem~\ref{thm:lb} holds when $n\ge \tilde\Omega_q(\frac{1}{\varepsilon^2})$.
    Hence, similar to Remark~\ref{rem:ub-eps}, the lower bound of Theorem~\ref{thm:lb} holds for $\varepsilon$ approaching 0 with $n$, so long as $\varepsilon\ge \Omega_q(\frac{1}{n^{1/2}})$.
  \end{remark}

  \begin{remark}
    \label{rem:code-avg}
    Theorem~\ref{thm:main} and Theorem~\ref{thm:lb} together show that the optimal number of traces for a code of rate $1-\varepsilon$ is bounded above and below by the number of traces for average-case trace reconstruction of a string of length $\poly(1/\varepsilon)$.
    More precisely, there exist $m_1=\tilde\Omega_q(\frac{1}{\sqrt{\varepsilon}})$ and $m_2 = \tilde O_q(\frac{1}{\varepsilon})$ such that the optimal number of traces for rate $1-\varepsilon$ coded trace reconstruction with failure probability $\frac{1}{3}$ is between $T_{q}\rand(m_1)$ and $O_q(\log\frac{1}{\varepsilon}) \cdot T_{q}\rand(m_2)$.
    Hence any qualitative improvement to the upper or lower bounds for coded trace reconstruction implies an analogous improvement for average-case trace reconstruction and vice versa.
  \end{remark}
  
  \paragraph{Large alphabet upper and lower bounds.}
  So far, we have focused on codes for binary alphabets. By defining the deletion channel for strings over larger alphabets in the same way as the binary deletion channel, one can ask questions for coded trace reconstruction over larger alphabets. In this setting, our results are stronger in two ways. Firstly, we are able to show matching upper and lower bounds for large alphabet trace reconstruction. Secondly, these constructions are simpler and do not rely on average-case trace reconstruction results.
  \begin{theorem}
    \label{thm:large-ub}
    For all $q,\varepsilon\in(0,1)$ and infinitely many $n$, there exists a rate $1-\varepsilon$ code over an alphabet of size $2^{O(\frac{1}{\varepsilon}\log\frac{1}{\varepsilon})}$ that is $(T,q,\delta)$ trace reconstructible for $T=O(\log_{1/q} \frac{1}{\varepsilon})$ and $\delta = 2^{-\Omega(n)}$ and which is encodable in time $O(n)$ and decodable in time $O(nT)$.
  \end{theorem}
  And as the following lower bound shows, this is tight in terms of the number of traces.
  \begin{theorem}
    \label{thm:large-lb}
    Any code (over any alphabet) of rate $1-\eps$ is not $(\floor{\log_{1/q}\frac{1}{\varepsilon}}, q,o(1))$ trace reconstructible.
  \end{theorem}
  We do not know if the dependence on $\varepsilon$ for the alphabet size in Theorem~\ref{thm:large-ub} is optimal. We leave understanding the trade-off between alphabet size and number of traces as an open question for future work.

\subsection{Techniques}

In this section we describe our constructions. We first combine synchronization strings \cite{HS17} and erasure codes \cite{GI05} to give our large alphabet construction (Theorem~\ref{thm:large-ub}), and match this construction with a simple lower bound (Theorem~\ref{thm:large-lb}).

Extending these ideas to our binary code construction (Theorem~\ref{thm:main}) requires more work, and we introduce a novel technique for binary code concatenation, turning our large alphabet code from Theorem~\ref{thm:large-ub} into a binary code. This concatenation also leverages codes for the binary deletion channel (e.g. \cite{GuruswamiL19}), and bounds for average-case trace reconstruction \cite{HPP18}.

We finish this section by describing our lower bound for coded trace reconstruction for the binary alphabet (Theorem~\ref{thm:lb}).
Trace reconstruction lower bounds usually find a hard pair of strings and prove that it takes many traces to distinguish these strings.
Coded trace reconstruction can simply avoid these hard pairs of strings, which makes applying prior results difficult.
Using techniques from information theory, we are able to transfer average-case trace reconstruction lower bounds to the coded setting.

\paragraph*{Large alphabet construction and lower bound.}
As a warm-up, first observe that any binary code $C\subset \{0,1\}^n$ can be turned into a code $C'$ over an alphabet of size $2n$ by mapping each codeword $(r_1,\dots,r_n)$ to a codeword $((r_1,1),(r_2,2),\dots,(r_n,n))\in(\{0,1\}\times [n])^n$.
This code has very low rate, but has the useful property that the deletion channel is essentially turned into an erasure channel: from a received string, we can always recover the indices of the received symbols, and thus the corresponding $r_i$.
If $C$ is a code of rate $1-\varepsilon$ tolerating a $\delta = \poly(\varepsilon)$ fraction of erasures, $C'$ is recoverable from $O(\log_{1/q} \frac{1}{\eps})$ traces: with high probability at most $q^T < \delta$ fraction of symbols are never received, producing less than $\delta n$ erasures, which can be corrected.

Our construction for large alphabets (Theorem~\ref{thm:large-ub}) uses the above intuition, but relies on synchronization strings to avoid ruining the rate of the resulting code. Instead of specifying the exact position of each symbol, we include a symbol of a synchronization string \cite{HS17} from a much smaller alphabet of size $\poly\left(\frac{1}{\varepsilon}\right)$.
We take our starting code $C$ to be over a large alphabet of size $2^{O(\frac{1}{\varepsilon}\log\frac{1}{\varepsilon})}$ and tolerate a $\delta=\poly (\varepsilon)$ fraction of erasures \cite{GI05}. Increasing the size of the alphabet beyond that of \cite{GI05} helps ensure the correct rate when combining with the synchronization string.
At the cost of a few more erasures, we can convert the outputs on the deletion channel into outputs with erasures and correct the erasures.

For the lower bound (Theorem~\ref{thm:large-lb}), any code of rate $1-\eps$ recovering from $T$ traces must also be able to recover from the erasure channel with erasure probability $q^T$, which has capacity at most $1 - q^T$. Therefore, $1-\eps < 1-q^T$ so $\log_{1/q} \frac{1}{\eps}$ traces are necessary for the erasure channel, and thus the deletion channel.

\paragraph*{Binary alphabet construction.}  Our construction for binary alphabets (Theorem~\ref{thm:main}) uses additional ideas beyond those in the large alphabet construction. 
Again, we use a high rate error correcting code with codewords $(r_1,\dots,r_{n_{out}}) \in C$ and a synchronization string $(s_1,\dots,s_{n_{out}})$.
Naively, one might ``concatenate'' the large alphabet construction with a high rate code of length $n_{in}=O(\frac{1}{\varepsilon}\log \frac{1}{\varepsilon})$ recoverable from a $O_\varepsilon(1)$ number of traces (which exists by \cite{HPP18}), so that each pair $(r_i,s_i)$ is encoded in a binary string $a_i$ of length $n_{in}$, and the final codeword is the concatenation $a_1||\cdots||a_{n_{out}}$.
Then, to recover the message, we first use the $T$ traces of the codeword $a_1||\cdots||a_{n_{out}}$ to recover $T$ traces of each $a_i$.
As in \cite{Cheraghchi:2019vd}, we can make sure we know where the traces of the $a_i$ start and finish by adding buffers of long runs on the ends of each $a_i$.
From the traces of each $a_i$, we run the inner trace reconstruction to recover each $a_i$, and thus recover the pair $(r_i,s_i)$.
We then run the outer error correction to fix any incorrectly decoded $r_i$'s.

This construction does not work for a subtle reason.
Because the length of each $a_i$ is a constant, we expect a (very small) constant fraction of the $a_i$'s buffers to be deleted, and we also expect a (very small) constant fraction of $a_i$'s to have deletions applied so that the interior of the $a_i$ looks like a buffer (we call this a ``spurious'' buffer).
From the $T$ traces of the codeword, we try to recover $T$ traces of each of the $a_i$'s using the buffers, but these $T$ traces, supposedly of $a_i$, might contain some traces of, e.g., $a_{i-5}$ or $a_{i+3}$.
Therefore, we need to know the synchronization symbols $s_i$ to determine which substrings of each of the $T$ traces belong to which $a_i$.
Thus, recovering the synchronization symbols must happen \emph{before} running trace reconstruction on the $a_i$'s.
However, the synchronization symbols $s_i$ are encoded in the $a_i$, so in this construction the synchronization symbols cannot be recovered until \emph{after} the trace reconstruction.

To avoid this issue, our construction crucially encodes the content symbol $r_i$ and the synchronization symbol $s_i$ separately. To our knowledge, this kind of concatenation has not appeared in other constructions of deletion codes.
Each content symbol $r_i$ is encoded using a high rate code of length $n_R=\Theta(\frac{1}{\varepsilon}\log\frac{1}{\varepsilon})$ obtained from bounds on average-case trace reconstruction. 
Each synchronization symbol is encoded in a code of length $n_S=\Theta(\log\frac{1}{\varepsilon})$ decodable in, crucially, 1 trace from the binary deletion channel. We can afford a very low rate code for the synchronization symbols because they are over a much smaller alphabet than the content symbols. Furthermore, we structure the encoded content symbols and encoded synchronization symbols so that they are not easily confused with each other. 

For the final decoding algorithm, we first recover the synchronization symbols within each trace. We then use the synchronization strings to determine the parts of each trace that corresponding to traces of a particular $a_i$. We then use these traces of $a_i$ in trace reconstruction to recover each content symbol $r_i$.
Finally, we use the error correction of the outer code $C$ to fix any mistakes in this process.

\paragraph*{Binary alphabet lower bound.}
Our binary lower bound (Theorem~\ref{thm:lb}) reduces coded trace reconstruction to constructing a code over an appropriately chosen memoryless channel, i.e. a channel where each alphabet symbol is corrupted independently or the other symbols.
In particular, we partition the input string $x \in \{0, 1\}^n$ into $n/m$ substrings of length $m \approx 1/\sqrt{\eps}$.
We then upper bound the rate of a code $C\subset (\{0,1\}^{m})^{n/m}$ over alphabet $\{0,1\}^m$ recovering a sequence $x$ of length $m$ substrings from $T=T\rand_{q}(m)$ independent traces of each of the $n/m$ substrings.
This is easier than recovering $x$ from $T$ independent traces of itself, so any rate upper bound for the code for $n/m$ substrings yields a rate upper bound for the original coded trace reconstruction problem.

Now, we can view the problem as coding over a discrete memoryless channel: we view our binary code as a code of length $n/m$ over the input alphabet $\mathcal{X}=\{0, 1\}^m$ and the channel produces outputs in $\mathcal{Y}=(\{0,1\}^*)^T$, corresponding to $T$ independent traces of the elements of $\mathcal{X}$.
By Shannon's noisy channel coding theorem~\cite{Shannon48}, the capacity of this channel equals the maximum, over distributions $\lambda$ on $\mathcal{X}$, of the mutual information $I(X_\lambda,Y_\lambda)$, where $X_\lambda\in \mathcal{X}$ is sampled from $\lambda$ and $Y_\lambda\in \mathcal{Y}$ is a tuple of $T$ strings each sampled as an independent trace of $X_\lambda$. 
Thus, to upper bound the rate of $C$, it suffices to upper bound the mutual information $I(X_\lambda,Y_\lambda)$ for all distributions $\lambda$ on $\mathcal{X}$.
If the distribution $\lambda$ is ``far'' from the uniform distribution, we can upper bound the mutual information by the entropy of $X_\lambda\sim \lambda$. Otherwise, if $\lambda$ is ``close'' to the uniform distribution, the mutual information is limited by the performance of average-case trace reconstruction. In either case, we get an upper bound on the mutual information which implies an upper bound on the rate of a code correctable from $T$ traces.

\subsection{Paper organization}

In Section~\ref{sec:prelim}, we define a few building blocks for our work. These include synchronization strings, codes for the binary deletion channel, and high rate error correcting codes. 
In Section~\ref{sec:large}, we present the proofs of our coded trace reconstruction results over large alphabets in Theorems~\ref{thm:large-ub} and Theorem~\ref{thm:large-lb}.
These proofs are simpler and serve as warm-ups for our results over binary alphabets, which require additional ideas.
In Section~\ref{sec:bin-sketch}, we sketch the proof of Theorem~\ref{thm:main}, showing how to convert upper bounds for average-case trace reconstruction into upper bounds for coded trace reconstruction.
In the remainder of Section~\ref{sec:binary-ub}, we formally prove Theorem~\ref{thm:main}.
In Section~\ref{sec:binary-lb}, we prove Theorem~\ref{thm:lb}, giving a black-block reduction from lower bounds for average-case trace reconstruction to lower bounds for coded trace reconstruction. Appendix~\ref{app:missing} fills in various technical details omitted from the main body.

\section{Preliminaries}
\label{sec:prelim}

\subsection{Basics}

All logs and exps are base 2 unless otherwise specified.
For an alphabet $\Sigma$, we let $\Sigma^*$ denote the set of strings over $\Sigma$ of any length.
For strings $w,w'$, we let $ww'$ denote the concatenation of strings $w$ and $w'$.
We may also denote the concatenation by $w||w'$ for clarity.
For a string $w$ and integer $i$, let $w^i$ denote the string $ww\cdots w$ with $w$ repeated $i$ times. A \emph{substring} is a sequence of consecutive characters in a string.
A \emph{run} is a maximal substring of a string all of whose bits are the same.
A \emph{partial function} $f:A\nrightarrow B$ is a function from a subset of $A$ to $B$.
For $x\in(0,1)$, let $H(x)= -x\log x - (1-x)\log(1-x)$ denote the binary entropy function.

A \emph{code} $C$ of length $n$ over an alphabet $\Sigma$ is a subset of $\Sigma^n$.
The elements of $C$ are called \emph{codewords}, and $n$ is called the \emph{length} of the code.
If $|\Sigma|=2$, we say $C$ is a binary code.
The \emph{rate} of a code $C$ is defined to be $\frac{\log|C|}{n\log|\Sigma|}$.
A code may have an associated \emph{message set} $\mathcal{M}$ and \emph{encoding function} $\Enc:\mathcal{M}\to C$, which is an injective map from messages to codewords.
By default, $\mathcal{M}=\{1,\dots,|C|\}$.
A code is \emph{decodable under the $\BDC_q$ with failure probability $\delta$} if it is $(1,q,\delta)$ trace reconstructible.
To \emph{construct} a code means to produce a description of its encoding and decoding functions.
Given two codes $C_1\subset \Sigma_1^{n_1}$ and $C_2\subset \Sigma_2^{n_2}$ with $|\Sigma_1|\le |C_2|$, a \emph{concatenation} of $C_1$ and $C_2$ is a code $C\subset \Sigma_2^{n_1n_2}$ whose codewords are $\Enc_2(c_1)||\dots||\Enc_2(c_{n_1})$ where $c_1\cdots c_{n_1}\in C_1$, and where $\Enc_2:\Sigma_1\to C_2$ is a fixed injective map.

We use the following forms of the Chernoff bound (e.g.,~\cite{DBLP:books/daglib/0025902})
\begin{lemma}[Chernoff bound -- discrete]
  Let $X_1,\dots,X_n$ be independent and identically distributed random variables with mean $\mu$ supported on $\{0,1\}$ 
  Then, for $\delta\ge 0$, 
  \begin{align}
    \label{eq:cher-1}
    \Pr[X_1+\cdots+X_n \le (1-\delta)\cdot n\mu] &\le e^{-\frac{\delta^2}{2}\cdot n\mu}  \\
    \label{eq:cher-2}
    \Pr[X_1+\cdots+X_n \ge (1+\delta)\cdot n\mu] &\le e^{-\frac{\delta^2}{2+\delta}\cdot n\mu}   .
  \end{align}
\end{lemma}

\begin{lemma}[Chernoff bound -- continuous]
  Let $X_1,\dots,X_n$ be independent and identically distributed random variables with mean $\mu$ supported on $[0,1]$ 
  Then, for $\delta\ge 0$, 
  \begin{align}
    \label{eq:cher-4}
    \Pr[X_1+\cdots+X_n \ge (1+\delta)\cdot n\mu] &\le e^{-2\delta^2\cdot \mu^2n}.
  \end{align}
\end{lemma}

\subsection{Short codes from average-case trace reconstruction}\label{sec:avg-to-coded}

In this section, we show a connection between short codes for trace reconstruction and average-case trace reconstruction. We use this connection to construct short, high-rate, trace reconstructible codes, which are building blocks in our main result.

The current state of the art for the optimal number of traces for average-case trace reconstruction is due to Holden, Pemantle, and Peres \cite{HPP18}, who show the following bound on $T\rand_{q,\beta}(n)$.
  \begin{theorem}[\cite{HPP18}]
  \label{thm:hpp}\label{thm:avg-tr}
  For all $q\in(0,1)$ and $\beta\ge 1$, we have $T\rand_{q,\beta}(n) \le \exp(O_{q,\beta}(\log^{1/3}n))$.
  \end{theorem}
  Note that the paper \cite{HPP18} only states Theorem~\ref{thm:hpp} for failure probability $1/n$, but their proof works in the same way for any polynomial failure probability $1/n^{\beta}$.
  There is also a slick way to amplify the failure probability in average-case trace reconstruction: with polynomially more traces, we can turn failure probability $1/n$ into $1/n^{\beta}$, by appending random bits to each trace and running trace reconstruction for $n'=n^\beta$ (see e.g., Theorem 3.2 of \cite{BCSS2019b}).

We now have the following two simple observations that results for average-case trace reconstruction show the existence of codes for coded trace reconstruction and vice versa.
\begin{claim}
  If there exists a code of size $2^n(1-o(1))$ that is $(T,q,o(1))$ trace reconstructible, then average case trace reconstruction can be done in $T$ traces with failure probability $o(1)$.
\end{claim}
 \begin{proof}
   The probability that a random string is both in the code and is decoded correctly from $T$ traces is at least $1-o(1)$.
 \end{proof}
 And conversely,
\begin{lemma}
  \label{lem:avg-tr}
  Let $\beta>1$, $q\in(0,1)$, and $T=T\rand_{q,2\beta}(n)$.
  For all positive integers $n$, there exists a code $C$ with $|C|\ge (1-n^{-\beta})2^n$ that is $(T,q,n^{-\beta})$ trace reconstructible.
\end{lemma}
\begin{proof}
  For any string $x\in \{0,1\}^n$, let $\delta_x$ denote the probability that $x$ is recovered incorrectly using the algorithm solving trace reconstruction for random traces on the $\BDC_q$ in $T$ traces with failure probability $n^{-2\beta}$.
  By definition of $T=T\rand_{q,2\beta}(n)$, we have $\E_x[\delta_x] \le  \frac{1}{3}n^{-2\beta}$, so, by Markov's inequality, $\Pr_x[\delta_x \ge n^{-\beta}] < n^{-\beta}$.
  Setting $C$ to be the set of all $x$ with $\delta_x\le n^{-\beta}$ and using the same trace reconstruction algorithm gives that $C$ is $(T,q,n^{-\beta})$-trace reconstructible, and has at least $(1-n^{-\beta})2^n$ codewords.
\end{proof}

We need to combine these short trace reconstruction codes into a longer one in Theorem~\ref{thm:main}. The following notion helps prevent these short codes from being confused with the other components of our construction.
\begin{definition}
  \label{def:protect}
  A string $w$ is \emph{$m$-protected} if it can be written as $w=0^{m}w^\circ 1^{m}$, where $w^\circ$ starts with a 1, ends with a 0, and every substring of $w^\circ$ of length $m'\ge m/4$ has between $\frac{m'}{4}$ and $\frac{3m'}{4}$ 1s (inclusive).
  In any $m$-protected string $w$, we let $w^\circ$ denote the string $w$ with the leading $m$ 0s and the trailing $m$ 1s deleted.
  We refer to $w^\circ$ as the \emph{interior} of $w$.
  A code is $m$-protected if all of its codewords are $m$-protected.
\end{definition}

We use short codes which are both $m$-protected and trace reconstructible in our construction. The following Lemma (see Appendix~\ref{app:short} for details) shows that these codes exist.
\begin{lemma}
  For all $q\in(0,1)$ and $\beta \ge 150$, there exists an absolute constant $\varepsilon_0 = \varepsilon_0(\beta,q)  > 0$ such that the following holds.
  For all $\varepsilon\in(0,\varepsilon_0)$ and $n\ge 8\beta\frac{1}{\varepsilon}\log\frac{1}{\varepsilon}$, if $m = \floor{\beta\log n}$ and $T = T\rand_{q,6\beta}(n)$, there exist codes of length $n$ and rate at least $1-\frac{\varepsilon}{2}$ that are $m$-protected and $(T,q,n^{-3\beta})$ trace reconstructible.
  \label{lem:short}
\end{lemma}

\subsection{Synchronization strings}

Synchronization strings \cite{HS17} are useful tools for turning synchronization errors (insertions and deletions) into erasures (replacing symbol with the symbol `?') and substitution errors (replacing symbol with another symbol).
Here, we state the construction of synchronization strings that we use and a few useful properties.

\newcommand{\ID}{\operatorname{ID}}

\begin{definition}[Insertion-deletion distance]
Given two strings $S \in \Sigma^n$ and $T \in \Sigma^m$, the \emph{insertion-deletion distance} between $S$ and $T$, denoted $\ID(S, T)$ is the minimum number of characters that needed to be inserted into $S$ and deleted from $S$ to produce $T$. 
\end{definition}

Insertion-deletion distance is similar to \emph{edit distance} which allows for substitutions at a cost of $1$. Observe that if $S$ and $T$ have disjoint character sets, then $\ID(S, T)$ is the sum of their lengths. 

\begin{definition}[$\eta$-synchronization string]
  String $S\in \Sigma^n$ is an $\eta$-synchronization string if for every $1\le i < j < k \le n+1$, we have that $\ID(S[i,j), S[j,k)) > (1-\eta)(k-i)$.
\end{definition}
\begin{theorem}[Theorems 4.5 and 4.7 of \cite{HS18}]
  \label{thm:ss}
  For any $\eta\in(0,1)$ and all $n$, one can construct an $\eta$-synchronization string of length $n$ in time $\poly(n)$ over an alphabet of size $6000\eta^{-4}$.
\end{theorem}

We now describe some useful properties of synchronization strings.
Informally, a \emph{string matching} between two strings describes how to transform one string into the other via insertions and deletions. 
We use a definition of string matching equivalent to the one introduced in \cite{HS17}.
\newcommand{\del}{\text{del}}
\begin{definition}[String matching]
For strings $c$ and $c'$ of length $n$ and $n'$, respectively, a \emph{string matching} is a strictly increasing partial function $i^*:[n']\nrightarrow[n]$ such that, for all $j$ in the domain of $i^*$, we have $c_{i^*(j)} = c'_{j}$.
Given a string matching, an index $j\in[n']$ is called \emph{successfully transmitted} if it is in the domain of $i^*$, and is called an \emph{insertion} otherwise.
An element $i\in[n]$ is called a \emph{deletion} if it is not in the codomain of $i^*$.
\end{definition}

A \emph{$(n,\delta)$-indexing algorithm for a string $S$} takes as input a string $S'$ of length $n'$ with an unknown string matching $i^*:[n']\nrightarrow[n]$ having at most $n\delta$ insertions and deletions and outputs an index in $[n]\cup\{\perp\}$ for every index in $[n']$.
We say the algorithm \emph{decodes index $j\in [n']$ correctly} under a string matching $i^*$ if it outputs $i^*(j)$ for index $j$ when $i^*(j)$ exists and outputs $\perp$ if it does not exist.
A \emph{misdecoding} of an algorithm is a successfully transmitted, incorrectly decoded index $j\in[n']$.
An indexing algorithm is \emph{error free} if every $j \in [n']$ is correctly decoded or is assigned $\perp$.

Haeupler and Shahrasbi proved many results showing that synchronization strings yield indexing algorithms with few misdecodings.
In this work, we use the following two results.

\begin{theorem}[Theorem 5.10 of \cite{HS17}]
Let $S$ be an $\eta$-synchronization string of length $n$.
Then there exists an $(n,\delta)$-indexing algorithm for $S$ guaranteeing at most $\frac{2n\delta}{1-\eta}$ misdecodings.
Furthermore, this algorithm runs in time $O(n^4)$
\label{thm:ss-2}
\end{theorem}

\begin{theorem}[Theorem 6.18 of \cite{HS17}]
  Let $S$ be an $\eta$-synchronization string of length $n$.
  There exists a linear time error-free deletion-only $(n,\delta)$-indexing algorithm for $S$ guaranteeing at most $\frac{\eta}{1-\eta}\cdot n\delta$ misdecodings.
\label{thm:ss-3}
\end{theorem}

\subsection{Binary deletion channel codes}

The following lemma gives codes for the BDC$_q$ with failure probability at most $\delta$ and length $O(\log \delta^{-1})$.
In our application, we take $\delta=\poly\frac{1}{\varepsilon}$, where $1-\varepsilon$ is the rate of our code.
A similar construction appears in \cite{GuruswamiL19} (Proof of Theorem 1). 
We provide a proof in Appendix~\ref{app:bdc} for completeness. 
\begin{lemma}
  \label{lem:bdc}
  For all $q\in(0,1)$ and positive integers $K$ and $m$, there exists a binary code $C:[2^K]\to \{0,1\}^{3Km}$ where every codeword has exactly $2K$ runs, all of which have length either $m$ or $2m$ and decodable in linear time under the $\BDC_q$ with failure probability at most $6K\cdot 2^{-(1-q)m/20}$.
\label{thm:ss-4}
\end{lemma}

\begin{remark}
  The code above has rate $\frac{1}{3m}$, which approaches 0 as $m$ grows. Using a construction similar to \cite{GuruswamiL19}, if we drop the requirement of runs having length exactly $m$ or $2m$, it is possible to achieve a failure probability $2^{-\Omega(m)}$ with a code of rate $c(1-q)$ for some absolute $c>0$. We use the result in Lemma~\ref{lem:bdc} as the proof is simpler and the result is sufficient.
\end{remark}

\subsection{High rate error correcting codes}
Our constructions leverage high rate (rate $1-\varepsilon$) error correcting codes that are polynomial time encodable and decodable from a $\poly(\varepsilon)$ fraction of worst-case substitution errors.
For the details of these constructions and their parameters, see Appendix~\ref{app:ecc}.

For our binary upper bound, it suffices to use the following variant of a construction by Justesen \cite{Justesen72}.
Conveniently, it gives codes for \emph{all} sufficiently large $n$, rather than only infinitely many $n$.
This property is necessary for Remark~\ref{rem:ub-eps}, where we wish to take $\varepsilon\to 0$ as $n\to\infty$ in our binary upper bound construction.
\begin{proposition}
  \label{prop:justesen}
  For every $\varepsilon\in(0,\frac{1}{2})$ and $\Sigma$ whose size is a power of 2, there exists an $n_0=\tilde \Theta(\frac{1}{\varepsilon^2})$ such that, for all $n\ge n_0$, there exists a code of length $n$ over alphabet $\Sigma$ of rate $1-\varepsilon$ that is encodable and decodable in time $O_\varepsilon(n^2)$ from up to a fraction $\frac{\varepsilon^2}{500\log\frac{1}{\varepsilon}}$ of worst-case substitution errors.
\end{proposition}

It would suffice to use Proposition~\ref{prop:justesen} for our large alphabet construction (Theorem~\ref{thm:large-ub}) as well.
However, using the following error correcting code of Guruswami and Indyk \cite{GI05} allows linear time encoding/decoding of our large alphabet construction.
\begin{proposition}
  \label{prop:gi}
  For every $\varepsilon\in(0,\frac{1}{2})$ and $\Sigma$ whose size is a power of 2, there exist an infinite family of codes over $\Sigma$ of rate $1-\varepsilon$ encodable in linear time and decodable in linear time from up to a fraction $\frac{1}{40}\varepsilon^3$ of worst-case substitution errors.
\end{proposition}

\section{Optimal number of traces for large alphabet codes}
\label{sec:large}

We begin by describing the upper and lower bounds for coded trace reconstruction over a large alphabet. Many of the tools used in this section are important building blocks for the analysis of coded trace reconstruction over a binary alphabet.

\subsection{Upper bound}
\begin{proof}[Proof of Theorem~\ref{thm:large-ub}]
  We start by defining a few parameters for our construction.

  \textbf{Parameters.}
  Let $T = \ceil{\log_{1/q} \frac{160}{\varepsilon^3}}$.
  Let $q' = \frac{1+q}{2}$ and $\eta=\frac{\varepsilon^3}{160T}$.
  Let $\Sigma_S$ be an alphabet such that there exist $\eta$-synchronization strings over $\Sigma_S$, and assume $|\Sigma_S|$ is a power of 2.
  We may take $|\Sigma_S|=O_q(\poly\frac{1}{\varepsilon})$ by Theorem~\ref{thm:ss}.

  \textbf{Code.}
  Let $C_1$ be a length $n$ erasure code over an alphabet $\Sigma_C$ of size $|\Sigma_S|^{\ceil{2/\varepsilon}}$, rate at least $1-\frac{\varepsilon}{2}$, and decodable from a $\frac{\varepsilon^3}{40}$ fraction of worst-case substitution errors, given by Proposition~\ref{prop:gi}.
  Let $s_1,s_2,\dots,s_n$ be an $\eta$-synchronization string over alphabet $\Sigma_S$.
  Let $\Sigma = \Sigma_C \times \Sigma_S$.
  Let $C$ be a code with encoding $\mathcal{M}\to \Sigma^n$ whose codewords are $(c_1,s_1),\dots,(c_n,s_n)$ for codewords $(c_1,\dots,c_n)\in C$.

  \textbf{Decoding algorithm.}
  For $t\in[T]$, let $z\ind{t} = (x_1\ind{t},y_1\ind{t}),\dots,(x_{n\ind{t}}\ind{t},y_{n\ind{t}}\ind{t})$ be the $t$th trace, which has length $n\ind{t}$.
  Call a trace $z\ind{t}$ for $t\in[T]$ \emph{useful} if $n\ind{t}\ge (1-q')\cdot n$.
  \begin{enumerate}
  \item For every useful trace $z\ind{t}$, run the error-free deletion-only $(n, q')$-indexing algorithm in Theorem~\ref{thm:ss-3} to obtain indices $i_1\ind{t},\dots,i_{n\ind{t}}\ind{t}\in[n]\cup\{\perp\}$.
  \item For $i=1,\dots,n$, if there exists a useful $t\in[T]$ and index $j\in [n\ind{t}]$ such that $i_j\ind{t} = i$, then let $\hat c_i = x_j\ind{t}$. Otherwise, let $\hat c_i = \perp$.
  \item Run the erasure decoding for $C_1$ on the string $(\hat c_1,\dots,\hat c_n)$ to obtain a message in $\mathcal{M}$.
  \end{enumerate}

  \textbf{Efficiency.} 
    The code $C_1$ and synchronization string can each be constructed in polynomial time.
    Since $C_1$ has linear time encoding, so does our code.
    Decoding takes time $O(n\log \frac{1}{\varepsilon})$:
    the indexing algorithm for synchronization strings takes linear time by Theorem~\ref{thm:ss} and we run it $T$ times, and decoding the code $C_1$ from the resulting erasures takes linear time by Proposition~\ref{prop:gi}.

  \textbf{Rate.}
  The rate of the code $C_1$ is at least $1-\frac{\varepsilon}{2}$, so there are $|\Sigma_C|^{n(1-\frac{\varepsilon}{2})} = |\Sigma|^{n(1-\frac{\varepsilon}{2})\cdot \frac{\log|\Sigma_C|}{\log|\Sigma|}} \ge |\Sigma|^{n(1-\varepsilon)}$ codewords.
  The inequality follows as $\frac{\log|\Sigma_C|}{\log|\Sigma|} > 1 - \frac{\varepsilon}{2}$
  Hence, the rate of $C$ is at least $1-\varepsilon$.

  \textbf{Analysis.}
  First, the probability that some trace is not useful is equal to the probability that a binomial $B(n,1-q)$ is at most $(1-q')n = \frac{1-q}{2}n$, which, by the Chernoff bound, is at most $e^{-(1-q)n/8}$.
  Thus, the probability that there exists a trace that is not useful is, by the union bound, at most $T\cdot e^{-(1-q)n/8}\le 2^{-\Omega(n)}$.

  For all useful $t\in[T]$, $z\ind{t}$ is obtained from applying at most $q' n$ deletions to $c$.
  Thus, the $(n,q')$ indexing-algorithm in Theorem~\ref{thm:ss-3} succeeds with at most $\frac{\eta}{1-\eta}\cdot nq' < 2\eta n$ misdecodings.
  Hence, for all $j\in[n\ind{t}]$, we either have $i_j\ind{t} = \perp$ or $j$ is correctly decoded, in which case $x_j\ind{t} = c_{i_j}$.
  We conclude that, for all $i=1,\dots,n$, we either have $\hat c_i = c_i$ or $\hat c_i=\perp$.
  We now simply need to lower bound the number of $\hat c_i$ that are not $\perp$.
  If every trace is useful, for each index $i$ with $\hat c_i =\perp$, either $(c_i,s_i)$ is deleted in every trace or some trace has a misdecoding at the image of $(c_i,s_i)$.
  The expected number of symbols $(c_i,s_i)$ deleted in every trace is $q^Tn$, so by the Chernoff bound~\ref{eq:cher-2}, the probability that there are more than $2q^Tn$ symbols deleted in every trace is $2^{-\Omega_q(n)}$.
  Across all traces, the total number of misdecodings is at most $T\cdot 2\eta n$ by above.
  Thus, with probability at least $1-2^{-\Omega_q(n)}$, there are at most $2q^Tn+2T\eta n < \frac{\varepsilon^3}{40}n$ indices $i$ with $\hat c_i = \perp$.
  Hence, as the code $C_1$ tolerates $\frac{\varepsilon^3}{40}\cdot n$ errors (and thus erasures), we decode our message correctly. 
\end{proof}

\subsection{Lower bound}

\newcommand{\EC}{\textnormal{EC}}
\newcommand{\DC}{\textnormal{DC}}
\begin{proof}[Proof of Theorem~\ref{thm:large-lb}]
For brevity, let $\DC_q$ denote the deletion channel with deletion probability $q$.
Let $\EC_q$ denote the erasure channel with erasure probability $q$. 
That is $\EC_q$ takes an input string and independently with probability $q$ replaces each symbol with the symbol `?'.

We show that a $(T,q,o(1))$ trace reconstructible code over the $\DC_q$ is a code for $\EC_{q^T}$ with block error probability $o(1)$.
To do this, we show that we can turn an output of $\EC_{q^T}$ into $T$ independent outputs of $\DC_q$.
From a single symbol sent over $\EC_{q^T}$, one can produce $T$ independent copies of the symbol sent across $\EC_q$: if the output is an erasure, return $T$ erasures, and if the output is the original symbol, return the output of $T$ independent copies of the symbol over $\EC_q$, conditioned on not all outputs being erasures.
Using the above, from a single output from $\EC_{q^T}$, one symbol at a time, produce $T$ independent outputs over $\EC_q$, and replace the erasures with deletions to obtain $T$ independent outputs over $\DC_q$, as desired.
Since the capacity of $\EC_{q^T}$ is $1-q^T$ (see e.g. \cite{Shannon48}), we have that our code cannot be $(T,q,o(1))$ trace reconstructible when $1-\varepsilon > 1- q^T$, i.e. $T < \log_{1/q}\frac{1}{\varepsilon}$.
\end{proof}

\section{Upper bound on traces for binary codes}
\label{sec:binary-ub}

In this section, we prove Theorem~\ref{thm:main}.

\subsection{Proof sketch}\label{sec:bin-sketch}

As the proof of Theorem~\ref{thm:main} is involved, we start with a sketch of the proof.
Throughout this proof sketch, fix $q$ to be some constant between 0 and 1.
We prove Theorem~\ref{thm:main} when $n$ is any sufficiently large multiple of a constant (the constant is $n_R+n_S = \Theta_q(\frac{1}{\varepsilon}\log\frac{1}{\varepsilon})$ in the proof).
To extend to all sufficiently large $n$ we simply pad the beginning of codewords in an existing code with 0s.

The proof uses concatenation on top of the construction for Theorem~\ref{thm:large-ub}.
Recall that the code in Theorem~\ref{thm:large-ub} is obtained by ``zipping'' codewords $r_1,\dots,r_n\in\Sigma_R^n$ from a high-rate error correcting code with a fixed synchronization string $s_1,\dots,s_n\in\Sigma_S^n$, where $|\Sigma_R|\ge |\Sigma_S|^{\Omega(1/\varepsilon)}$.
We call the elements of $\Sigma_R$ \emph{content symbols} and the elements of $\Sigma_S$ \emph{synchronization symbols}.

\textbf{A first attempt.}
Naively, we could concatenate the code in Theorem~\ref{thm:large-ub} of rate $1-\Theta(\varepsilon)$ over the large alphabet $\Sigma_R\times \Sigma_S$ with binary code $C_R$ with encoding $\Enc_{R}:\Sigma_R\times \Sigma_S\to \{0,1\}^{n_{R}}$ of length $n_{R} = \tilde \Theta(\frac{1}{\varepsilon})$ and rate $1-\Theta(\varepsilon)$ that is decodable from $T=\exp(O(\log^{1/3}\frac{1}{\varepsilon}))$ traces, giving a concatenated code of rate $1-\Theta(\varepsilon)$ (such a code exists by \cite{HPP18}).
In this way, the binary codewords are of the form $\Enc_R(r_1,s_1)||\cdots||\Enc_R(r_n,s_n)$.
Then, perhaps, from $T$ traces, we could run the trace reconstruction algorithm for $C_R$ to recover guesses $(\hat r_i,\hat s_i)$ for $(r_i,s_i)$, and then run the outer decoding to correct any errors/insertions/deletions.

\textbf{The problem and the fix.}
The problem with the above approach is that we need to recover the synchronization information of the inner codewords \emph{before} we run the inner trace reconstruction algorithm: we do not know, for instance, where the trace of $\Enc_{R}(r_1,s_1)$ ends and the trace of $\Enc_{R}(r_2,s_2)$ starts.
To fix this, we need the following key idea: separately encode the content symbol $r_i$ and the synchronization symbol $s_i$.
Further, in order to ensure that the encoded content bits and the encoded synchronization bits are not confused, we ensure that (1) the encoded synchronization bits only have long runs and (2) the encoded content bits are relatively dense in both 0s and 1s in every small interval (with the exception of one long run at the beginning and end of the string).
This yields the encoding of $(r_1,s_1),\dots,(r_n,s_n)$ depicted below.
\newcommand\encodeDiagram{
  \underbrace{
    \underbrace{0\ldots0}_{m\text{-bit buffer}}
    \underbrace{\text{interior of $a_1$}}_{n_{R}-2m\text{ bits}}
    \underbrace{1\ldots1}_{m\text{-bit buffer}}
  }_{a_1=\Enc_R(r_1)}
  \:\Big|\Big|\:
  \underbrace{
    \underbrace{0\ldots0}_{k_1\in\{m,2m\}}\underbrace{1\ldots1}_{\ell_1}\cdots \underbrace{0\ldots0}_{k_K}\underbrace{1\ldots1}_{\ell_K}
  }_{\text{$b_1=\Enc_S(s_1)$: $2K$ long runs}}
  \:\Big|\Big|\:
  \cdots
  \:\Big|\Big|\:
  \Enc_R(r_n)
  \:\Big|\Big|\:
  \Enc_S(s_n)
}
\begin{align}
  \encodeDiagram
\end{align}

\textbf{Code construction sketch.}
We take our outer error correcting code $C_{out}$ to have length $n_{out}$, rate $1-\Theta(\varepsilon)$, tolerate a $\Theta(\varepsilon^3)$ fraction of worst-case substitution errors, and with alphabet $\Sigma_R$ equal in size to the number of codewords in $C_R$, which we may take to be a power of two by arbitrarily throwing out at most half of the codewords.
Such a code exists by Proposition~\ref{prop:gi}.
We think of $n_{out}$ as growing and all other parameters as fixed.
We take a synchronization string $s_1,\dots,s_{n_{out}}$ with constant synchronization parameter $\eta=\Theta(1)$.
We take the length of the encoding $\Enc_R(r_i)$ of $r_i$ to be $\Theta(\frac{1}{\varepsilon}\log\frac{1}{\varepsilon})$, and the length of the encoding $\Enc_S(s_i)$ to be $\Theta(\log\frac{1}{\varepsilon})$, so the rate is at least $1-\varepsilon$.
We ensure that the encoding of $r_i$ is recoverable from $T$ traces with failure probability at most $O(\varepsilon^{100})$.
We also ensure that each encoded word of $r_i$ is $m$-protected in the sense of Definition~\ref{def:protect}.
The average-case trace reconstruction results of Holden, Pemantle, and Peres \cite{HPP18} implies that such a code exists (see Lemma~\ref{lem:short}).
We also ensure that the encoding of $s_i$ is recoverable from \emph{one} trace of the $\BDC_q$ with failure probability $\varepsilon^{100}$. 
Note that, since the synchronization parameter $\eta$ is a constant, we have that $|\Sigma_S|$ is a constant, so such a code $C_S$ with encoding $\Enc_S:\Sigma_S\to \{0,1\}^{\Theta(\log(1/\varepsilon))}$ for the $\BDC_q$ exists (see Lemma~\ref{lem:bdc}).

\textbf{Decoding algorithm sketch.}
Our decoding algorithm is depicted in Figure~\ref{fig:2} and divides into three steps.
\begin{enumerate}
\item (Trace alignment) For each $t\in[T]$, for all $i\in[n_{out}]$, determine an estimate $\widehat{\tau\ind{t}(a_i)}$ for the bits from the $i$th content symbol 
\item (Inner trace reconstruction) For $i\in[n_{out}]$, run the trace reconstruction for the code $C_R$ on $\widehat{\tau\ind{1}(a_i)},\dots,\widehat{\tau\ind{T}(a_i)}$ to recover an estimate for $\hat r_i$.
\item (Outer error correction) Run the error correction for $C_{out}$ on the estimates $\hat r_1,\dots,\hat r_{n_{out}}$.
\end{enumerate}
\begin{figure}
  \label{fig:2}
  \begin{center}
    \usetikzlibrary{calc}
    \begin{tikzpicture}
      \tikzstyle{caption}=[font=\small,text width=6.2em,minimum height=1em,align=left]
      \tikzstyle{block}=[font=\small,text width=4em,minimum height=1em,align=center]
      \tikzstyle{error1}=[fill=red,fill opacity=0.4]
      \tikzstyle{error2}=[fill=red,fill opacity=0.4]
      \tikzstyle{error3}=[fill=blue,fill opacity=0.4]
      \foreach \i in {0,...,6} {
        \foreach \j in {-1,...,6} {
          \coordinate (\i-\j) at (\i*2,-\j);
        }
      }
      \foreach \j in {0,...,6} {
          \coordinate (-1-\j) at (-2.5,-\j);
      }

      \node[caption] at (-1-0) (00) {Step 1 $(t=1)$:};
      \node[block,draw] at (0-0) (00) {$\widehat{\tau\ind{1}(a_1)}$};
      \node[block,draw,error1] at (1-0) (10) {$\widehat{\tau\ind{1}(a_2)}$};
      \node[block     ] at (2-0) (20) {$\cdots$};
      \node[block,draw] at (3-0) (30) {$\widehat{\tau\ind{1}(a_i)}$};
      \node[block     ] at (4-0) (40) {$\cdots$};
      \node[block,draw] at (5-0) (50) {$\widehat{\tau\ind{1}(a_{n_{out}})}$};

      \node[caption] at (-1-1) (00) {Step 1 $(t=2)$:};
      \node[block,draw] at (0-1) (01) {$\widehat{\tau\ind{2}(a_1)}$};
      \node[block,draw] at (1-1) (11) {$\widehat{\tau\ind{2}(a_2)}$};
      \node[block     ] at (2-1) (21) {$\cdots$};
      \node[block,draw,error2] at (3-1) (31) {$\widehat{\tau\ind{2}(a_i)}$};
      \node[block     ] at (4-1) (41) {$\cdots$};
      \node[block,draw] at (5-1) (51) {$\widehat{\tau\ind{2}(a_{n_{out}})}$};

      \node[caption] at (-1-2) (00) {Step 1 $(t=2)$:};
      \node[block,draw] at (0-2) (02) {$\widehat{\tau\ind{3}(a_1)}$};
      \node[block,draw] at (1-2) (12) {$\widehat{\tau\ind{3}(a_2)}$};
      \node[block     ] at (2-2) (22) {$\cdots$};
      \node[block,draw] at (3-2) (32) {$\widehat{\tau\ind{3}(a_i)}$};
      \node[block     ] at (4-2) (42) {$\cdots$};
      \node[block,draw] at (5-2) (52) {$\widehat{\tau\ind{3}(a_{n_{out}})}$};

      \node[block     ] at (0-3) (03) {$\vdots$};
      \node[block     ] at (1-3) (13) {$\vdots$};
      \node[block     ] at (2-3) (23) {$\vdots$};
      \node[block     ] at (3-3) (33) {$\vdots$};
      \node[block     ] at (4-3) (43) {$\vdots$};
      \node[block     ] at (5-3) (53) {$\vdots$};

      \node[caption] at (-1-4) (00) {Step 1 $(t=T)$:};
      \node[block,draw] at (0-4) (04) {$\widehat{\tau\ind{T}(a_1)}$};
      \node[block,draw] at (1-4) (14) {$\widehat{\tau\ind{T}(a_2)}$};
      \node[block     ] at (2-4) (24) {$\cdots$};
      \node[block,draw] at (3-4) (34) {$\widehat{\tau\ind{T}(a_i)}$};
      \node[block     ] at (4-4) (44) {$\cdots$};
      \node[block,draw] at (5-4) (54) {$\widehat{\tau\ind{T}(a_{n_{out}})}$};

      \coordinate (x1) at ($(50.north west)+(-0.2,0.15)$);
      \coordinate (x2) at ($(50.north east)+(0.2,0.15)$);
      \coordinate (x3) at ($(54.south east)+(0.2,-0.15)$);
      \coordinate (x4) at ($(54.south west)+(-0.2,-0.15)$);
      \draw[thick,blue]  (x1) -- (x3);
      \draw[thick,blue]  (x2) -- (x4);

      \node[caption    ] at (-1-5) (-25) {Step 2:};
      \node[block,draw] at (0-5) (05) {$\hat r_1$};
      \node[block,draw,error1] at (1-5) (15) {$\hat r_2$};
      \node[block     ] at (2-5) (25) {$\cdots$};
      \node[block,draw,error2] at (3-5) (35) {$\hat r_i$};
      \node[block     ] at (4-5) (45) {$\cdots$};
      \node[block,draw,] at (5-5) (55) {$\hat r_{n_{out}}$};
      \draw[very thick,blue] (55.south west) -- (55.north east);
      \draw[very thick,blue] (55.south east) -- (55.north west);

      \node[caption    ] at (-1-6) (-16) {Step 3:};
      \node[draw,text width=25em,align=center] at (5,-6) (m) {Output in $\mathcal{M}$};
      \draw[thick] (-4.5,-4.55) -- (11.5,-4.55);

      \draw[->] (05.south) -- (m);
      \draw[->] (15.south) -- (m);
      \draw[->] (35.south) -- (m);
      \draw[->] (55.south) -- (m);

    \end{tikzpicture}
  \end{center}
  \caption{Decoding: Inner trace reconstruction and outer error correction.
  Index $i$ is incorrect only if (i) some trace $t$ incorrectly guesses the image of $a_i$ (shaded red), or (ii) the inner trace reconstruction procedure $\Dec_{R}$ fails (X-ed out blue).}
\end{figure}
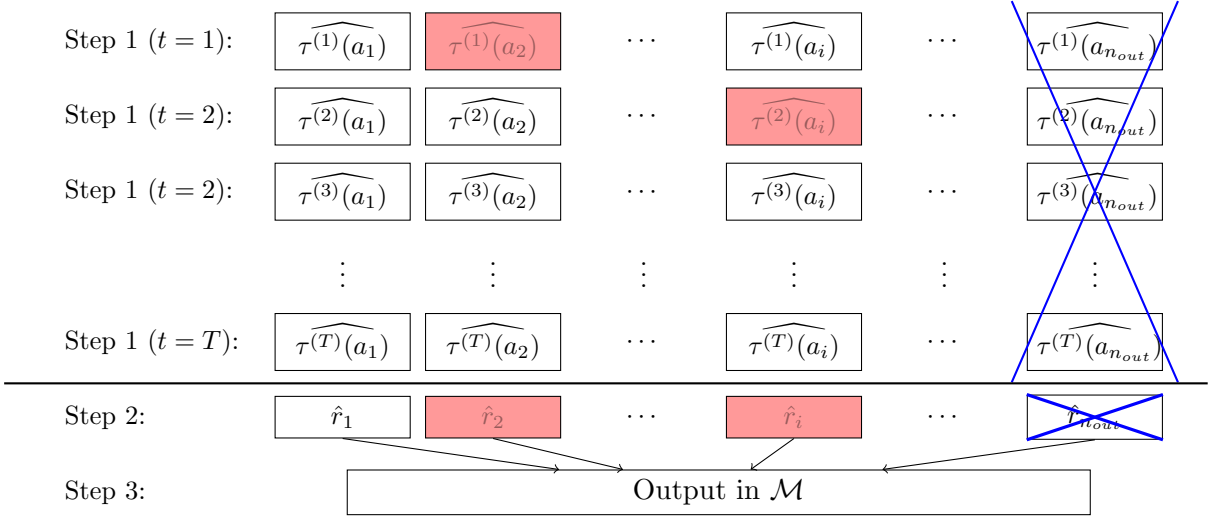

\textbf{Decoding analysis sketch.}
Let $a_i = \Enc_{R}(r_i)$ and $b_i = \Enc_S(s_i)$ be the binary encodings of the $i$th content symbol and $i$th synchronization symbol, respectively.
We call $a_i$ a \emph{content block} and $b_i$ a \emph{synchronization block}.
For $t\in[T]$ and $i\in[n_{out}]$, let $\tau\ind{t}(a_i)$ and $\tau\ind{t}(b_i)$ denote the images of the $i$th content symbol and $i$th synchronization symbol, respectively, in the $t$th trace.

The key to the analysis is that, by using the indexing algorithm for synchronization strings, with high probability for every trace $t$, the estimates of all but a $O(\varepsilon^{100})$ fraction of the images $\widehat{\tau\ind{t}(a_i)}$ are \emph{exactly} correct, i.e. satisfy $\tau\ind{t}(a_i) = \widehat{\tau\ind{t}(a_i)}$.
This is because, with high probability, in each trace, we can find $(1-O(\varepsilon^{100}))n_{out}$ pairs of strings $(x_j\ind{t},y_j\ind{t})$ equal to some $(\tau\ind{t}(a_i),\tau\ind{t}(b_i))$ by simply scanning the trace.
Then, from the substrings $\tau\ind{t}(b_i)$, we can recover a $1-O(\varepsilon^{100})$ fraction of the synchronization symbols $s_i$.
Using the synchronization symbols, we run the indexing algorithm for the synchronization string to match the pairs $(x_j\ind{t},y_j\ind{t})$ to the correct index $i\in[n_{out}]$, so that, in each trace, $1-O(\varepsilon^{100})$ fraction of the pairs are indexed correctly.
This produces $(1-O(\varepsilon^{100}))n_{out}$ accurate estimates $\widehat{\tau\ind{t}(a_i)}$ in every trace.

If the above holds, by the union bound, for all but a $O(T\cdot \varepsilon^{100})\le O(\varepsilon^{99})$ (recall $T=\exp(O_q(\log^{1/3}(\frac{1}{\varepsilon})))=\varepsilon^{-o(1)}$, and assume $\varepsilon$ is sufficiently small) fraction of indices $i\in[n_{out}]$, the image of the inner codeword $a_i$ is correctly determined in \emph{every trace}.
Among these indices $i\in[n_{out}]$, we expect the inner trace reconstruction algorithm to fail on a $O(\delta_{R})$ fraction, and for the rest of these indices,
\begin{align}
  \hat r_i 
  = \Dec_{in}(\widehat{\tau\ind{1}(a_i)},\dots,\widehat{\tau\ind{T}(a_i)})
  = \Dec_{in}(\tau\ind{1}(a_i),\dots,\tau\ind{T}(a_i))
  = r_i.
\end{align}
Thus, the fraction of indices $i\in[n_{out}]$ for which $\hat r_i\neq r_i$ is $O(\gamma T + \delta_{R}) = O(\varepsilon^{99})$ with high probability, which, by our choice of parameters, is less than the fraction of substitution errors tolerated by our outer code.
Hence, the outer error correction succeeds with high probability, as desired.
We point out that all of the big-Os in the outer error fractions do not have a dependence on $q$, so the parameters of the outer code do not need to depend on $q$.

\subsection{Construction}

First we define the code.

\textbf{Parameters}.\footnote{We make no attempt to optimize the constants in the proof.}
Let $\beta=\frac{10^4}{(1-q)^3}$, let $\tilde\varepsilon_0 = \tilde \varepsilon_0(\beta, q)$ be given by Lemma~\ref{lem:short}.
Let $n_{R} \defeq \floor{10^4\beta\frac{1}{\varepsilon}\log(\frac{1}{\varepsilon})}$. This is the length of our inner codeword (``R'' for ``reconstruction'').
Let $T  \defeq T\rand_{q,6\beta}(n_{R})$. This is the number of traces we use.
Let $\delta_{R} \defeq n_{R}^{-3\beta}$. This is an upper bound on the failure probability of the inner code's reconstruction algorithm.
By Theorem~\ref{thm:avg-tr}, $T \le \exp(O(\log^{1/3}n_{R})) < \varepsilon^{o(1)}$.
Thus, for $\varepsilon$ sufficiently small, we have (i) $T < \frac{1}{\varepsilon}$, (ii) $\varepsilon < \beta^{-1}$, and (iii) $\varepsilon < \tilde \varepsilon_0$. 
For the rest of the proof, assume $\varepsilon$ is such that all three items hold.

Let $m \defeq \floor{\beta\log n_{R}}$. This is the size of a ``buffer''.
Let $m'\defeq \frac{1}{2}(1-q)m$. This is the threshold for deciding whether a run in a trace is interpreted as a buffer or not.
Let $\eta \defeq \frac{1}{3}$ be the synchronization parameter.
Let $K\defeq 20$. This is the number of bits in a synchronization symbol.
Let $n_S \defeq 60m$. This is the number of bits in an encoded synchronization symbol.
Let $\delta_S\defeq 6K \cdot 2^{-(1-q)m/40}$. This is an upper bound on the probability a synchronization symbol is decoded correctly.
Let $\gamma \defeq 2^{-(1-q)m/80}$. This is an upper bound on the probability an inner codeword is ``incorrectly parsed'' (defined below).
In this way, $\gamma T < \varepsilon^{100}$.
Let $\delta_{out}=\frac{1}{50000}\varepsilon^3$ be a bound on the outer code's error tolerance.
Throughout this analysis we think of $q,\varepsilon,\beta,m,m',K,n_R,n_S, \delta_{R}, \delta_S, \delta_{out}, T, \gamma,\eta$ as fixed, and $n_{out}$, the length of the outer code defined below, as growing.

\begin{figure}
\begin{center}
\def\arraystretch{1.2}
  \begin{tabular}[width=\textwidth]{|c|c|c|p{9cm}|}
    \hline
    Parameter  & Value & $\lim_{\varepsilon\to 0}$ & Description \\
    \hline
    $q$ & & & Deletion probability  \\\hline
    $\varepsilon$ & & & Constructed code has rate $1-\varepsilon$  \\\hline
    $\beta$ & $\frac{10^4}{(1-q)^3}$ & $\Theta_q(1)$ & Large constant  \\\hline
    $n_{R}$ & $\floor{10^4\beta\frac{1}{\varepsilon}\log(\frac{1}{\varepsilon})}$ & $\tilde \Theta_q(\frac{1}{\varepsilon})$ & Content block length  \\\hline
    $T$ & $T\rand_{q,6\beta}(n_{R})$ & $\varepsilon^{-o_q(1)}$ & Number of traces used \\\hline
    $\delta_{R}$ & $n_{R}^{-3\beta}$ & $O(\varepsilon^{100})$ &  Upper bound on inner code $C_R$'s trace reconstruction failure probability \\\hline
    $m$ & $\floor{\beta\log n_{R}}$  & $\Theta_q(\log \frac{1}{\varepsilon})$ & Size of a buffer \\\hline
    $m'$ & $\frac{1}{2}(1-q)m$ & $\Theta_q(\log\frac{1}{\varepsilon})$ & Threshold for interpreting an output run as a buffer \\\hline
    $\eta$    & $\frac{1}{3}$ & $\Theta(1)$ & Synchronization parameter \\\hline
    $K$ & $20$ & $\Theta(1)$ & Number of bits in synchronization symbol: $|\Sigma_S| = 2^K$ \\\hline
    $n_S$ & $60m$ & $\Theta_q(\log\frac{1}{\varepsilon})$ & Number of bits in encoded synchronization symbol \\\hline
    $\delta_S$ & $6K \cdot 2^{-(1-q)m/40}$ & $O(\varepsilon^{100})$ & Upper bound on inner code $C_S$'s decoding failure probability\\\hline
    $\gamma $ & $2^{-(1-q)m/80}$ & $O(\varepsilon^{100})$ & Upper bound on probability content block is ``incorrectly parsed'' \\\hline
    $n_{out}$ & $\to \infty$ & $\to\infty$ & Outer code length \\ \hline
    $\delta_{out}$ & $\frac{1}{50000}\varepsilon^3$ & $\Omega(\varepsilon^3)$ & Lower bound on the outer code's error tolerance \\\hline
  \end{tabular}
\end{center}
\end{figure}

\textbf{Inner codes.}
By our choice of parameters, $\beta\ge 150$, $\varepsilon < \tilde \varepsilon_0$, $n_{R}\ge 8\beta\frac{1}{\varepsilon}\log\frac{1}{\varepsilon}$, and $T=T\rand_{q,6\beta}(n_R)$.
By Lemma~\ref{lem:short} there exists a code $C_R$ of length $n_R$ and rate $1-\frac{\varepsilon}{2}$ with message set $\Sigma_R$ and encoding function $\Enc_{R}:\Sigma_R\to \{0,1\}^{n_R}$ all of whose codewords are $m$-protected (as $m=\floor{\beta \log n_{R}}$), and that is $(T,q,\delta_{R})$ trace reconstructible.
By removing at most half of the codewords (arbitrarily), we may assume the alphabet size $|\Sigma_R|$ is a power of 2, and the rate is at least $1-\frac{\varepsilon}{2} - \frac{1}{n_{R}}$.
Let the corresponding decoding function be $\Dec_R:(\{0,1\}^*)^T\to \Sigma_R$.

Let $\eta=\frac{1}{3}$, and let $s_1,\dots, s_n$ be an $\eta$-synchronization string of length $n$ over alphabet $\Sigma_S$ of size $2^{K}$: such strings exist by Theorem~\ref{thm:ss} and are constructible in polynomial time.
By Lemma~\ref{lem:bdc}, there exists a code $C_S$ with encoding function $\Enc_S:\Sigma_S\to \{0,1\}^{n_S}$ and decoding function $\Dec_S:\{0,1\}^*\to \Sigma_S$ that is decodable under the BDC$_q$ with failure probability at most $\delta_S$, all of whose codewords start with a 0, end with a 1, and have runs of length exactly $m$ or $2m$.

\textbf{Outer code.}
Let $C_{out}:\mathcal{M}\to \Sigma_R^{n_{out}}$ be a code of length $n_{out}$ and rate $1-\frac{\varepsilon}{10}$ over the alphabet $\Sigma_R$ correcting a $\frac{(\frac{\varepsilon}{10})^2}{500\log\frac{10}{\varepsilon}} > \delta_{out}$ fraction of worst-case errors, given by Proposition~\ref{prop:justesen}.

\textbf{Encoding.}
Our encoding is as follows. 
\begin{enumerate}
\item Take a message and encode it with $C_{out}$ to obtain symbols $r_1,\dots,r_n\in\Sigma_R$.
\item Let $a_i = \Enc_R(r_i)$ and $b_i=\Enc_S(s_i)$. We call $a_i$ a \emph{content block} and $b_i$ a \emph{synchronization block}.
\item Concatenate $c=a_1||b_1||a_2||b_2||\cdots||a_n||b_n$.
\end{enumerate}
\begin{align}
\encodeDiagram
\end{align}

\textbf{Length and Rate.}
The length is $n_{out}\cdot (n_{R}+n_S)$.
The outer code rate is $(1-\frac{\varepsilon}{10})$, the inner code $C_R$ rate is $1-\frac{\varepsilon}{2} - \frac{1}{n_{R}}$, and the synchronization symbols multiply the rate by $1-\frac{n_S}{n_{R}+n_S} > 1 - \frac{60\beta\log n_{R}}{10^4\beta\frac{1}{\varepsilon}\log\frac{1}{\varepsilon}} = 1 - \frac{\varepsilon}{10}$.
The total rate is thus at least $(1-\frac{\varepsilon}{10})(1-\frac{\varepsilon}{2} - \frac{1}{n_R})(1-\frac{\varepsilon}{10}) > 1-\varepsilon$.

\textbf{Decoding.}
Let $z\ind{1},\dots,z\ind{T}$ be the traces.
We use the following notation for the ``Trace Alignment'' step of the decoding below.
The crucial elements of the Trace Alignment step's analysis are given in Definition~\ref{def:parse}, Lemma~\ref{lem:good-2}, and Lemma~\ref{lem:good-1}.
For every trace, call a (maximal) run of length greater than $m'$ a \emph{decoded buffer}.
Call every bit in a decoded buffer a \emph{decoded buffer bit}, and call all other bits \emph{decoded content bits}.
For every trace, define a \emph{decoded content block} to be a substring of the form $0^{t_0}w1^{t_1}$, where the first $t_0$ 0s and the last $t_1$ 1s each form a decoded buffer, and $w$ is a nonempty string of decoded content bits. Note in particular that $w$ must start with a $1$ and end with a $0$. If two decoded content blocks overlapped, say $0^{t_0}w1^{t_1}$ and $0^{t_0'}w'1^{t_1'}$, then $0^{t_0}$ is the same run of bits as $0^{t_0'}$, because $w$ does not consist of any decoded buffers. Likewise, $1^{t_1}$ is the same run of bits as $1^{t_1'}$ so $w = w'$. Therefore, any two decoded content blocks are disjoint.

We can thus enumerate the decoded content blocks of a trace $t$ in order $x_1\ind{t},\dots,x_{n\ind{t}}\ind{t}$, where $n\ind{t}$ is the number of decoded content blocks in trace $t$.
For each decoded content block $x_j\ind{t}$, we define the associated \emph{decoded synchronization block} $y_j\ind{t}$ as the substring between $x_j\ind{t}$ and $x_{j+1}\ind{t}$ (or the end of the string, if $k=n\ind{t}$).\footnote{Note that there may be some bits at the beginning of the string that are neither in decoded content blocks nor in decoded synchronization blocks.}
Our decoding algorithm is as follows.

\begin{enumerate}
\item (Trace alignment) For each trace $t\in[T]$, compute $\widehat{\tau\ind{t}(a_1)},\dots,\widehat{\tau\ind{t}(a_{n_{out}})}$ as follows:
  \begin{enumerate}
    \item Compute the decoded content blocks $x\ind{t}_1,\dots,x\ind{t}_{n\ind{t}}$ of $z\ind{t}$ along with their associated decoded synchronization blocks $y\ind{t}_1,\dots,y\ind{t}_{n\ind{t}}$.
    \item For all $j\in[n\ind{t}]$, decode a synchronization symbol $\hat s_j\ind{t} \defeq \Dec_S(y\ind{t}_{j})\in\Sigma_S$ from the decoded synchronization block. 
    \item From the string $\hat s_1\dots \hat s_{n\ind{t}}$, obtain indices $\hat i_1\ind{t},\dots,\hat i_{n\ind{t}}\ind{t}$ using the $(n,13\gamma)$ indexing algorithm in Theorem~\ref{thm:ss-2}.
    \item For $j=1,\dots,n\ind{t}$, let $\widehat{\tau\ind{t}(a_{\hat i_j})}\defeq x\ind{t}_{j}$, and let $\widehat{\tau\ind{t}(a_i)} = \perp$ for $i\notin \{\hat i_1,\dots,\hat i_{n\ind{t}}\}$.
    Here, $\widehat{\tau\ind{t}(a_i)}$ is a string denoting our guess for the image of $a_i$ in the $t$th trace.
  \end{enumerate}
\item (Inner trace reconstruction) For $i=1,\dots,n$, let $\hat r_i = \Dec_{R}(\widehat{\tau\ind{1}(a_i)},\dots,\widehat{\tau\ind{T}(a_i)})\in \Sigma_R$.
\item (Outer error correction) Run $\Dec_{out}(\hat r_1,\dots,\hat r_n)$ to obtain a message in $\mathcal{M}$.
\end{enumerate}

\textbf{Run time.} 
The encoding runs in time $O(n^2)$: the outer encoding runs in time $O(n^2)$, and each of the $O(n)$ inner encodings runs in time $O_\varepsilon(1)$. 

The decoding runs in $O_\varepsilon(n^4)$ time.
Determining the decoded content blocks and decoded synchronization blocks can be done in linear time $O(T\cdot n_{out}) = O(n_R\cdot n_{out}) \le O(n)$.
Here, we used that $T < n_R$ as $\varepsilon$ is sufficiently small, so in particular, this decoding time has no dependence on $q$.
The code $C_S$ is decodable from deletions in linear time by Lemma~\ref{lem:bdc}, so computing all synchronization symbols $\tilde s_j\ind{t}$ takes $O(n)$ time.
The indexing step for the synchronization string takes time $O(n_{out}^4)$ in each trace, so all indexing steps take time $O(T\cdot n_{out}^4)\le O(n^4)$.
Thus, the entire trace alignment step takes time $O(n^4)$.
Each inner trace reconstruction step takes $n_R^{1+o(1)}$ time \cite{HPP18}, so the entire inner trace reconstruction step (Step 2) takes $n_R^{1+o(1)}\cdot n_{out} \le n^{1+o(1)}$ time by running the decoder for \cite{HPP18}.
The outer error correction (Step 3) runs in time $O(n^2)$.
The total decoding time is thus $O(n^4)$.

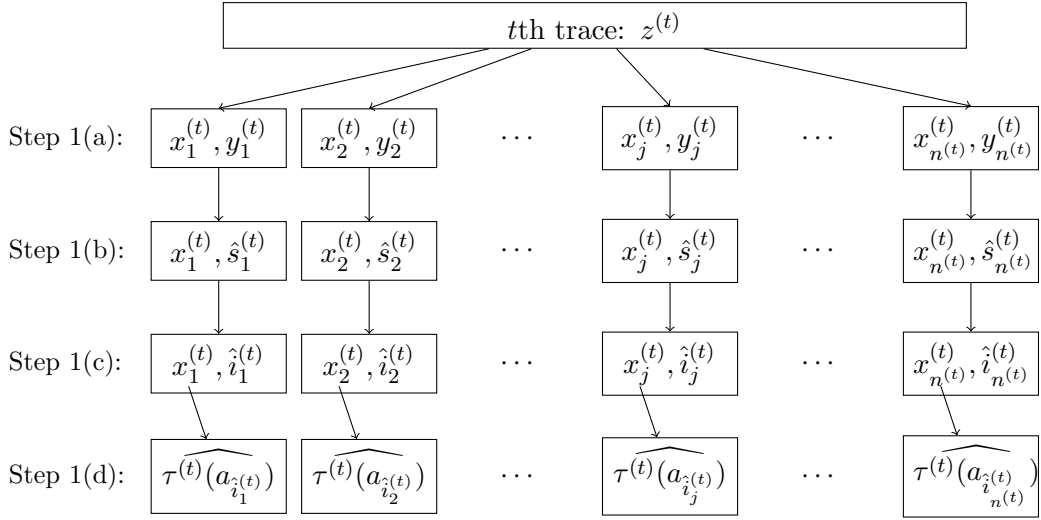
\begin{figure}
\begin{center}

\begin{tikzpicture}
  \tikzstyle{caption}=[font=\small,text width=4.1em,minimum height=1em,align=left]
  \tikzstyle{block}=[text width=4em,minimum height=1em,align=center]
  \foreach \i in {-1,...,6} {
    \foreach \j in {-1,...,6} {
      \coordinate (\i-\j) at (\i*2,-\j*1.5);
    }
  }
  \node[draw,text width=25em, align=center] at (5,1.5) (z) {$t$th trace: $z\ind{t}$};

  \node[caption   ] at (-1-0) (00) {Step 1(a):};
  \node[block,draw] at (0-0) (00) {$x_1\ind{t},y_1\ind{t}$};
  \node[block,draw] at (1-0) (10) {$x_2\ind{t},y_2\ind{t}$};
  \node[block     ] at (2-0) (20) {$\cdots$};
  \node[block,draw] at (3-0) (30) {$x_j\ind{t},y_j\ind{t}$};
  \node[block     ] at (4-0) (40) {$\cdots$};
  \node[block,draw] at (5-0) (50) {$x_{n\ind{t}}\ind{t},y_{n\ind{t}}\ind{t}$};

  \node[caption  ] at (-1-1) (-11) {Step 1(b):};
  \node[block,draw] at (0-1) (01) {$x_1\ind{t},\hat s_1\ind{t}$};
  \node[block,draw] at (1-1) (11) {$x_2\ind{t},\hat s_2\ind{t}$};
  \node[block,    ] at (2-1) (21) {$\cdots$};
  \node[block,draw] at (3-1) (31) {$x_j\ind{t},\hat s_j\ind{t}$};
  \node[block,    ] at (4-1) (41) {$\cdots$};
  \node[block,draw] at (5-1) (51) {$x_{n\ind{t}}\ind{t},\hat s_{n\ind{t}}\ind{t}$};

  \node[caption  ] at (-1-2) (-12) {Step 1(c):};
  \node[block,draw] at (0-2) (02) {$x_1\ind{t},\hat i_1\ind{t}$};
  \node[block,draw] at (1-2) (12) {$x_2\ind{t},\hat i_2\ind{t}$};
  \node[block,    ] at (2-2) (22) {$\cdots$};
  \node[block,draw] at (3-2) (32) {$x_j\ind{t},\hat i_j\ind{t}$};
  \node[block,    ] at (4-2) (42) {$\cdots$};
  \node[block,draw] at (5-2) (52) {$x_{n\ind{t}}\ind{t},\hat i_{n\ind{t}}\ind{t}$};

  \node[caption  ] at (-1-3) (-13) {Step 1(d):};
  \node[block,draw] at (0-3) (03) {$\widehat{\tau\ind{t}(a_{\hat i_1\ind{t}})}$};
  \node[block,draw] at (1-3) (13) {$\widehat{\tau\ind{t}(a_{\hat i_2\ind{t}})}$};
  \node[block,    ] at (2-3) (23) {$\cdots$};
  \node[block,draw] at (3-3) (33) {$\widehat{\tau\ind{t}(a_{\hat i_j\ind{t}})}$};
  \node[block,    ] at (4-3) (43) {$\cdots$};
  \node[block,draw] at (5-3) (53) {$\widehat{\tau\ind{t}(a_{\hat i_{n\ind{t}}\ind{t}})}$};

  \draw[->] (z) -- (00.north);
  \draw[->] (z) -- (10.north);
  \draw[->] (z) -- (30.north);
  \draw[->] (z) -- (50.north);

  \draw[->] (00) -- (01);
  \draw[->] (10) -- (11);
  \draw[->] (30) -- (31);
  \draw[->] (50) -- (51);

  \draw[->] (01) -- (02);
  \draw[->] (11) -- (12);
  \draw[->] (31) -- (32);
  \draw[->] (51) -- (52);

  \draw[->] (0-2)+(-0.4,-0.3) -- (03);
  \draw[->] (12)+(-0.4,-0.3) -- (13);
  \draw[->] (32)+(-0.4,-0.3) -- (33);
  \draw[->] (52)+(-0.4,-0.3) -- (53);
\end{tikzpicture}

\end{center}
\caption{Decoding: Trace alignment}
\label{fig:1}
\end{figure}

\subsection{Analysis}

\subsubsection{Notation}
Throughout this analysis, we think of all the bits of $c,z\ind{1},\dots,z\ind{t}$ as distinct.
We may informally refer to $z\ind{t}$ as ``trace $t$''.
Let $\tau\ind{t}$ denote the $t$th \emph{deletion pattern}, i.e.\ a map from bits in the codeword $c$ to bits in $t$th trace $z\ind{t}$.\footnote{Formally, $\tau\ind{t}$ is an injective and surjective partial function from the bits of $c$ to the bits of $z\ind{t}$, such that the undeleted bits of $c$ form the domain of $\tau\ind{t}$, and these bits are mapped in order to the bits of $z\ind{t}$.}
In this way, $\tau\ind{t}(c) = z\ind{t}$, and if $c'$ is a substring of $c$, then $\tau\ind{t}(c')$ is a substring of $z\ind{t}$.
Throughout the analysis, if $w$ and $w'$ are substrings of a trace $z\ind{t}$, we write $w=\ind{t}w'$ to indicate that they are the same substring of trace $z\ind{t}$.

\subsubsection{Correctly parsed indices}

\begin{definition}[Spurious buffer]
  We say that a \emph{spurious buffer} of a trace $z\ind{t}$ is a decoded buffer that is a substring of the image of a content block's interior, i.e. a substring of $\tau\ind{t}(a_i^\circ)$ for some $i$.
\end{definition}

\begin{definition}
  \label{def:parse}
  For $t\in[T], i\in\{0,\dots,n_{out}+1\}$, we say index $i$ is \emph{intact in trace $t$} if $i=0$, $i=n+1$, or all of the following hold:
  \begin{enumerate}
  \item At least $m'$ of the $m$ leading 0s of $a_i$ are not deleted in $\tau\ind{t}$
  \item At least $m'$ of the $m$ trailing 1s of $a_i$ are not deleted in $\tau\ind{t}$
  \item $\tau\ind{t}(a_i)$ has no spurious buffers. 
  \item The image of all runs of $b_i$ under $\tau\ind{t}$ have length at least $m'$. 
  \item $\Dec_{S}(\tau\ind{t}(b_i))\neq s_i$ 
  \end{enumerate}
  For $t\in[T]$, call an index $i\in\{1,\dots,n_{out}\}$ \emph{correctly parsed in trace $t$} if indices $i-1,i$, and $i+1$ are all intact in trace $t$, and \emph{incorrectly parsed in trace $t$} otherwise.
\end{definition}
Note that the event ``$i$ is intact in trace $t$'' depends only on the images of $a_i$ and $b_i$ under $\tau\ind{t}$, and hence all such events are independent.
The following lemma justifies the terminology in Definition~\ref{def:parse}.
\begin{lemma}
  If index $i\in[n_{out}]$ is correctly parsed in trace $t\in[T]$, there exists an index $j\in[n\ind{t}]$ such that $\tau\ind{t}(a_i)=\ind{t}x_j\ind{t}$, and $\tau\ind{t}(b_i)=\ind{t}y_j\ind{t}$.
  \label{lem:good-2}
\end{lemma}
\begin{proof}
  First, we show that, since indices $i-1$ and $i$ are intact in trace $t$, the image $\tau\ind{t}(a_i)$ of content block $a_i$ forms a decoded content block in $z\ind{t}$.
  The image of the substring $b_{i-1}||a_i||b_i$ of our codeword $c$ under the $t$h deletion pattern $\tau\ind{t}$ is $\tau\ind{t}(b_{i-1})||\tau\ind{t}(a_i)||\tau\ind{t}(b_i)$.
  Further, the substring $\tau\ind{t}(b_{i-1})$ ends in a 1 (property 2), the substring $\tau\ind{t}(a_i)$ begins with a 0 and ends with a 1 (properties 1 and 2), and the substring $\tau\ind{t}(b_i)$ starts with a 0 (property 1).
  Hence $\tau\ind{t}(a_i)$ starts with a decoded buffer of 0s (property 1), ends with a decoded buffer of 1s (property 2), and has no other decoded buffers (property 3), so $\tau\ind{t}(a_i)$ is a decoded content block.
  Thus, there exists some $j\in[n\ind{t}]$ such that $\tau\ind{t}(a_i) =\ind{t}x_j\ind{t}$.

  Similarly, since indices $i$ and $i+1$ are intact in trace $t$, the substring $\tau\ind{t}(a_{i+1})$ is a decoded content block by the same reasoning.
  Since index $i$ is intact in trace $t$, all the bits in the image of $b_i$ under $\tau\ind{t}$ are decoded buffer bits (property 4), so there are no decoded content bits between $\tau\ind{t}(a_i)$ and $\tau\ind{t}(a_{i+1})$.
  Thus, $\tau\ind{t}(a_i)$ and $\tau\ind{t}(a_{i+1})$ are consecutive decoded content blocks, so $\tau\ind{t}(b_i) =\ind{t} y_j\ind{t}$, as desired.
\end{proof}
As an immediate corollary, we have the following lemma.
\begin{lemma}
  Let $t\in[T]$.
  If there are at least $k$ indices $i\in[n_{out}]$ that are correctly parsed in trace $t$, then the sequences of pairs $(\tau\ind{t}(a_1),s_1),\dots,(\tau\ind{t}(a_n),s_n)$ and $(x_1\ind{t}, \hat s_1\ind{t}),\dots,(x_{n\ind{t}}\ind{t},\hat s_{n\ind{t}}\ind{t})$  have a common subsequence of length $k$.
  \label{lem:good-3}
\end{lemma}
\begin{proof}
  By Lemma~\ref{lem:good-2}, the sequences of pairs of strings $(\tau\ind{t}(a_1),\tau\ind{t}(b_1)),\dots,(\tau\ind{t}(a_n),\tau\ind{t}(b_n))$ and $(x_1\ind{t}, y_1\ind{t}),\dots,(x_{n\ind{t}}\ind{t},y_{n\ind{t}}\ind{t})$ have a common subsequence of length $k$, namely the subsequence corresponding to the $k$ correctly parsed pairs $(\tau\ind{t}(a_i),\tau\ind{t}(b_i))$.
  The result follows as $\Dec_{in}(\tau\ind{t}(b_i)) = s_i$ and $\Dec_{in}(y_j\ind{t}) = \hat s_j\ind{t}$, so we can apply the $\Dec_{in}$ operator to the second element in each pair of each sequence to obtain the desired result.
\end{proof}

\subsubsection{Bounding incorrectly parsed indices}
The following lemma guarantees that for all $t\in[T]$ and $i\in[n_{out}]$, the string $\tau\ind{t}(a_i)$ has no spurious buffers with high probability.
\begin{lemma}
  For any $t\in[T],i\in[n_{out}]$, the expected number of spurious buffers in $\tau\ind{t}(a_i)$ is at most $e^{-(1-q)m/40}$.
  \label{lem:good-0}
\end{lemma}
\begin{proof}
  Call a substring of a content block's interior $a_i^\circ$ \emph{spurious buffer indicator} if it has length exactly $\floor{\frac{m}{4}}$ and at least one of the following occur:
  \begin{enumerate}
  \item All of the 0s are deleted.
  \item All of the 1s are deleted.
  \item At least $m'$ of the 0s are not deleted.
  \item At least $m'$ of the 1s are not deleted.
  \end{enumerate}
  Consider a spurious buffer of 0s.
  Let $c'$ denote the minimal substring of $c$ containing the spurious buffer's preimage.
  In $c'$, all the 1s are deleted and at least $m'$ of the 0s are not deleted.
  Hence, any substring or superstring of this preimage $c'$ of length $\floor{\frac{m}{4}}$ is a spurious buffer indicator.
  The same holds for the preimage of a spurious buffer of 1s.
  Hence, we may identify each spurious buffer with a corresponding spurious buffer indicator of $a_i^\circ$ of length $\frac{m}{4}$ (breaking ties arbitrarily).
  Because distinct spurious buffers are disjoint, the spurious buffer indicator are distinct.
  Thus, the number of runs of decoded buffer bits is bounded above by the number of spurious buffer indicator.
  Since $a_i$ is $m$-protected, every substring of length $\floor{\frac{m}{4}}$ has at least $\floor{\frac{m}{16}}$ 0s and at least $\floor{\frac{m}{16}}$ 1s.
  Thus, the probability all the 0s are deleted is at most $q^{-\floor{m/16}} < e^{-(1-q)m/20}$, and the probability all the 1s are deleted is also at most $e^{-(1-q)m/20}$.
  Any run of $\floor{\frac{m}{4}}$ bits has at most $\floor{\frac{m}{4}}$ 0s, and the probability that at least $m'$ 0s are not deleted is bounded above by the probability that the binomial random variable $B(\floor{\frac{m}{4}}, 1-q)$ is at least $m'=\frac{m(1-q)}{2}$, which, by the Chernoff bound \eqref{eq:cher-2} is at most $e^{-(1-q)\floor{m/12}}$.
  The same probability holds for the 1s.
  Hence, the probability that any run of $\floor{\frac{m}{4}}$ bits is a spurious buffer indicator is at most $4e^{-(1-q)m/20}$, so the expected number of spurious buffer indicator, and thus the number of spurious buffers, is at most $4n_{R}\cdot e^{-(1-q)m/20}$ by linearity of expectation, which, by definition of $m$ and $n_{R}$, is at most $e^{-(1-q)m/40}$.
\end{proof}

\begin{lemma}
  \label{lem:good-1}
  For $t\in[T], i\in[n_{out}]$, the probability index $i$ is intact is at least $1-\gamma$.
\end{lemma}
\begin{proof}
  We bound the probability each property in Definition~\ref{def:parse} fails.
  \begin{enumerate}
  \item Since $a_i$ begins with $m = \frac{2m'}{1-q}$ leading 0s, the expected number of undeleted 0s among the leading 0s is distributed as the binomial $B(m,1-q)$, which has mean $2m'$.
  Hence, the probability that property 1 fails is at most probability that this binomial is less than $m'$, which, by the Chernoff bound \eqref{eq:cher-1}, is at most $e^{-m'/4}$.
  \item By the same reasoning, the probability that property 2 fails is at most $e^{-m'/4}$.

  \item If property 3 fails, $\tau\ind{t}(a_i)$ has a spurious buffer.
  Since $\Pr[X>0]\le \E[X]$ for all nonnegative integer random variables $X$, the probability $\tau\ind{t}(a_i)$ has a spurious buffer is at most $e^{-(1-q)/40}$ by Lemma~\ref{lem:good-0}.

  \item By construction, $b_j$ has $2K$ runs of length at least $m$.
  Property 4 fails if some run has less than $m'$ non-deleted bits.
  The number of non-deleted bits in a run is distributed as one of the binomials $B(m,1-q)$ or $B(2m,1-q)$.
  Hence, by the Chernoff bound \eqref{eq:cher-1} and union bound, property 4 fails with probability at most $2K \cdot e^{-m'/4}$.

  \item Property 5 fails with probability at most $\delta_S \le 6K\cdot 2^{-(1-q)m/40}$ by the definition of code $C_S$.
  \end{enumerate}
  By the union bound, the probability that index $i$ is not intact in trace $t$ is at most
  \begin{align}
    \underbrace{e^{-m'/4}}_{\text{property 1}}
    + \underbrace{e^{-m'/4}}_{\text{property 2}}
    + \underbrace{4n_{R}\cdot e^{-(1-q)m/12}}_{\text{property 3}}
    &+ \underbrace{2K\cdot e^{-m'/4}}_{\text{property 4}}
      + \underbrace{6K\cdot 2^{-(1-q)m/40}}_{\text{property 5}} \nonumber\\
    \ &\le \ (2+4n_{R}+8K) 2^{-(1-q)m/40}.
        \label{}
  \end{align}
  For our choice of $n_{R}$, we have $2+4n_{R}+8K < 2^{(1-q)m/80}$, so the probability index $i$ is not intact in trace $t$ is at most $2^{-(1-q)m/80} = \gamma$.
\end{proof}
A simple Chernoff bound gives the following corollary.
\begin{corollary}
  The probability that there exists a $t\in[T]$ with more than $6\gamma n_{out}$ incorrectly parsed indices is at most $2^{-\Omega(n_{out})}$.
  \label{cor:good-1}
\end{corollary}
\begin{proof}
  For $t\in[T]$ and $i\in[n_{out}]$, let $\mathcal{E}_{t,i}$ denote the event that index $i$ is not intact in trace $t$.
  The events $\mathcal{E}_{t,i}$ are all independent, and any such event happens with probability at most $\gamma$ by Lemma~\ref{lem:good-1}.
  Hence, by the Chernoff bound \eqref{eq:cher-2} and the union bound, the probability there exists a $t\in[T]$ with more than $2\gamma n_{out}$ events $\mathcal{E}_{t,i}$ occurring is at most $T\cdot 2^{-\gamma n_{out}/3}$.
  Hence, as the number of incorrectly parsed indices $i$ in a trace $t$ is at most 3 times the number of non-intact pairs, the probability that there exists a $t$ with more than $6\gamma n_{out}$ incorrectly parsed pairs is also at most $T\cdot 2^{-\gamma n_{out}/3} = 2^{-\Omega(n_{out})}$.
\end{proof}

\subsubsection{Most traces of content bits are recovered}
Note that a trace may have more than $n_{out}$ decoded content blocks if there are spurious buffers.
The next lemma shows that, with high probability, this does not happen too much.
\begin{lemma}
  \label{lem:cher-2}
  For any $t$, with probability $1-\exp(-\Omega(n_{out}))$, we have $n\ind{t} \le (1+\gamma)n_{out}$.
\end{lemma}
\begin{proof}
  The number of decoded blocks is bounded above by $n_{out}$ plus the number of spurious buffers.
  Let $X_i$ denote the number of spurious buffers in block $\tau\ind{t}(a_i)$.
  By the definition of spurious buffer, $X_1,\dots,X_{n_{out}}$ are all independent.
  Additionally, each block has at most $n_{R}$ bits, so it certainly has at most $n_{R}$ spurious buffers.
  Hence, $\frac{X_i}{n_{R}} \in[0,1]$ for all $i$.
  By the Lemma~\ref{lem:good-0}, $\E[\frac{X_i}{n_{R}}] \le \frac{e^{-(1-q)m/40}}{n_R} = \frac{\gamma^2}{n_{R}}$ for all $i=1,\dots,n$.
  Thus, by the Chernoff bound \eqref{eq:cher-4} on the variables $\frac{X_1}{n_{R}},\dots, \frac{X_{n_{out}}}{n_{R}}$, the probability that $X_1+\cdots+X_{n_{out}}\ge \gamma n_{out}$ is at most $2^{-\gamma^2 n_{out}/3} \le 2^{-\Omega(n_{out})}$.
\end{proof}

\begin{lemma}
  Let $t\in[T]$.
  If there are at least $(1-6\gamma)n_{out}$ correctly parsed indices in trace $t$ and if $n\ind{t}\le (1+\gamma)n_{out}$, then there are at least $(1-46\gamma)n_{out}$ indices $i$ such that $\tau\ind{t}(a_i) = \widehat{\tau\ind{t}(a_i)}$.
  \label{lem:health-3}
\end{lemma}
\begin{proof}
  Suppose there are at most $\gamma n_{out}$ incorrectly parsed indices in trace $t$ and also that $n\ind{t}\le (1+\gamma)n$.
  By Lemma~\ref{lem:good-3}, there is a common subsequence between $(\tau\ind{t}(a_1),s_1),\dots,(\tau\ind{t}(a_n),s_n)$ and $(x_1\ind{t}, \hat s_1\ind{t}),\dots,(x_{n\ind{t}}\ind{t},\hat s_{n\ind{t}})$ of length at least $(1-6\gamma)n$.
  Since $n\ind{t}\le (1+\gamma)n$, there exists a string matching between $(\tau\ind{t}(a_1),s_1),\dots,(\tau\ind{t}(a_n),s_n)$ and $(x_1\ind{t}, \hat s_1\ind{t}),\dots,(x_{n\ind{t}}\ind{t},\hat s_{n\ind{t}})$ with at most $6\gamma n$ deletions and at most $7\gamma n$ insertions, for a total of at most $13\gamma n$ insertions or deletions.
  In particular, there are at least $(1-6\gamma)n_{out}$ correctly transmitted indices.
  This string matching gives a corresponding string matching between $s_1,\dots,s_n$ and $\hat s_1\ind{t},\dots,\hat s_{n\ind{t}}\ind{t}$.
  In this string matching, for the correctly transmitted symbols $j\in[n\ind{t}]$, let $i_j\ind{t}$ be such that $\hat s_{j}=s_{i_j\ind{t}}$, so that $x_j\ind{t} = \tau\ind{t}(a_{i_j\ind{t}})$.
  By Theorem~\ref{thm:ss-2}, the $(n,13\gamma)$-indexing algorithm for the $\eta$-synchronization string $s_1,\dots,s_n$ guarantees that there are at most $\frac{2}{1-\eta}\cdot 13\gamma n < 40\gamma n$ misdecodings.
  That is, there are at most $40\gamma n$ correctly transmitted indices $j=1,\dots,n\ind{t}$ such that $\hat i_j\ind{t} \neq i_j\ind{t}$.
  For the other correctly transmitted indices $j\in n\ind{t}$, we have $\tau\ind{t}(a_{i_j\ind{t}}) = x_j\ind{t} = \widehat{\tau\ind{t}(a_{\hat i_j\ind{t}})}$.
  Hence, there are at least $(1-6\gamma)n - 40\gamma n = (1-46\gamma)n$ indices $i$ such that $\tau\ind{t}(a_i) = \widehat{\tau\ind{t}(a_i)}$.
\end{proof}

\subsubsection{Finishing the proof}
We now complete the proof.
The next lemma shows that, with high probability, most of the inner trace reconstructions succeed ``in theory''.
That is, they succeed assuming the trace alignment steps recovered the images of the $a_i$'s successfully in all traces.
\begin{lemma}
  \label{lem:cher-3},
  With probability $1-\exp(-\Omega(n_{out}))$, for all but at most $2\delta_{R}$ fraction of indices $i\in[n_{out}]$, we have $\Dec_{R}(\tau\ind{1}(a_i),\dots,\tau\ind{T}(a_i))=r_i$.
\end{lemma}
\begin{proof}
  Call an index $i$ \emph{incorrect} if $\Dec_{R}(\tau\ind{1}(a_i),\dots,\tau\ind{T}(a_i))\neq r_i$.
  The probability an index $i$ is incorrect is at most $\delta_{R}$ as $C_R$ is $(T,q,\delta_{R})$ trace reconstructible and $\tau\ind{1}(a_i),\dots,\tau\ind{T}(a_i)$ are independent traces of $a_i$.
  Furthermore, for all $i$, the events that $i$ is incorrect are independent of each other.
  The expected number of incorrect $i$ is at most $\delta_{R}n$, so by the Chernoff bound \eqref{eq:cher-2}, the probability the number of incorrect $i$ is larger than $2\delta_{R}n$ is at most $2^{-\delta_{R}n/3}$, as desired.
\end{proof}

Now we can prove Theorem~\ref{thm:main}.
\begin{proof}[Proof of Theorem~\ref{thm:main}]
  By Corollary~\ref{cor:good-1}, Lemma~\ref{lem:cher-2}, and Lemma~\ref{lem:cher-3}, with probability $1-2^{-\Omega(n)}$, all the following occur:
  \begin{enumerate}
  \item Every trace has at most $6\gamma n_{out}$ incorrectly parsed indices.
  \item For all $t\in[T]$, we have $n\ind{t} \le (1+\gamma)n_{out}$.
  \item All but at most $2\delta_{R}n_{out}$ indices $i\in[n_{out}]$ satisfy $\Dec_{R}(\tau\ind{1}(a_i),\dots,\tau\ind{T}(a_i))=r_i$.
  \end{enumerate}
  We show that, when all of the above occur, the decoding succeeds.
  By Lemma~\ref{lem:health-3} and properties 1 and 2 above, there are at least $(1-46\gamma)n_{out}$ indices $i$ such that $\tau\ind{t}(a_i) = \widehat{\tau\ind{t}(a_i)}$.
  Thus, there are at least $(1-46\gamma T)n_{out}$ indices $i$ such that $\tau\ind{t}(a_i) = \widehat{\tau\ind{t}(a_i)}$ for all $t\in[T]$.
  Hence, by property 3 above, there are at least $(1-46\gamma T - 2\delta_{R})n_{out} > (1-\delta_{out})n_{out}$ indices $i$ such that $\tau\ind{t}(a_i) = \widehat{\tau\ind{t}(a_i)}$ for all $t\in[T]$ and $\Dec_{R}(\tau\ind{1}(a_i),\dots,\tau\ind{T}(a_i))=r_i$.
  For all such indices $i$, we have
  \begin{align}
    \hat r_i
    = \Dec_{R}\left(\widehat{\tau\ind{1}(a_i)},\dots,\widehat{\tau\ind{T}(a_i)}\right)
    = \Dec_{R}\left(\tau\ind{1}(a_i),\dots,\tau\ind{T}(a_i)\right)
    = r_i.
  \end{align}
  As $C_{out}$ tolerates $\delta_{out}n_{out}$ errors, the outer decoding finds the correct message in $\mathcal{M}$, as desired.
\end{proof}

\subsection{Extending to all sufficiently large \texorpdfstring{$n$}{n}}
For some $n_0=\poly\frac{1}{\varepsilon}$, the above construction gives error probability $\delta<\frac{1}{3}$ for all $n\ge n_0$ that are multiples of $n_R+n_S= \Theta(\frac{1}{\varepsilon}\log\frac{1}{\varepsilon})$: synchronization strings exist for all lengths $n_{out}$, the codes in Proposition~\ref{prop:justesen} exist for all lengths $n_{out}$ at least $\Omega(\frac{1}{\varepsilon^3})$, and the overall error probability is bounded by $2^{-\Omega(\gamma^2n)}$, so it suffices to take $n\gg \Omega(\frac{1}{\gamma^2}) = \poly\frac{1}{\varepsilon}$.

To extend to larger values of $n$, we take our constructed code of length $n - (n\mod (n_R+n_S))$ and pad the beginning of all codewords with $(n\mod (n_R+n_S))$ 0s.
This multiplies the rate by at least $1 - \frac{n_R + n_S}{n} < 1 - o(\varepsilon)$, so the rate is still at least $1-\varepsilon$.
Lemma~\ref{lem:good-2} still holds for all indices $i\in[n_{out}]$ except possibly the first ($x_1\ind{t}$ may have some extra 0s padded to $\tau\ind{t}(a_1)$), so the number of incorrect parsed indices is now bounded by $6\gamma n_{out}+1$.
For $n$ (and thus $n_{out}$) a sufficiently large polynomial in $\frac{1}{\varepsilon}$, the total fraction of incorrect (outer) content symbols $\hat r_i \neq r_i$ is still less than $\delta_{out}$ with high probability, so the decoding still succeeds with high probability.

\section{Lower bound on traces for binary codes}
\label{sec:binary-lb}

In this section, we prove the following theorem, which implies Theorem~\ref{thm:lb}.
\begin{theorem}\label{thm:lb-bin}
  Let $q\in(0,1)$ and $\varepsilon<\frac{1}{4}$. 
  Let $m=\floor{\sqrt{\frac{1/\varepsilon}{128\log(1/\varepsilon)}}}$ and $T = T_{q,0}\rand(m) - 1$.
  Then, for all $\delta \in(0,1)$, there exists $n_0=O_{\delta}(1/\eps^{2})$ such that all rate $1-\varepsilon$ codes of length at least $n_0$ are not $(T, q, \delta)$-trace reconstructible.
\end{theorem}

\subsection{Mutual information and Shannon's theorem}

Recall that the entropy of a random variable $X$ is $H(X) \defeq -\sum_{x} \Pr[X = x] \log \Pr[X = x]$.
For two random variables $X$ and $Y$ their conditional entropy of $Y$ given $X$ is defined to be $H(X|Y)\defeq \sum_{y}^{} \Pr[Y=y]\cdot H(X|Y=y)$, where $H(X|Y=y)$ is the entropy of the random variable $X$ given that $Y=y$.
From this, we can define their mutual information $I(X, Y)$ to be $I(X, Y) \defeq H(X) - H(X | Y)$.
A \emph{discrete memoryless channel} has finite input alphabet $\mathcal{X}$ and finite output alphabet $\mathcal{Y}$, and is given by a matrix $w(y|x)$, denoting, for each $x\in \mathcal{X}$, a distribution over received symbols $y\in\mathcal{Y}$.
With $w$, any probability distribution over $\mathcal{X}$ gives a joint distribution on $\mathcal{X},\mathcal{Y}$.

Given a discrete memoryless channel $w$, we say a code $C\subset \mathcal{X}^n$ is \emph{decodable with failure probability at most $\delta$} if there exists a map $f:\mathcal{Y}^n\to \mathcal{X}^n$ such that, for all $x_1\cdots x_n\in C$, we have
\begin{align}
  \Pr_{y_i\sim w(\cdot|x_i)}\left[f(y_1,\dots,y_n)\neq x_1\cdots x_n\right] \le \delta.
\end{align}
We need the following result, which provides a strong converse to Shannon's noisy channel coding theorem \cite{Shannon48}.
\begin{theorem}[e.g. Theorem 3.3.1 of \cite{Wolfowitz78}]
  \label{thm:wolfowitz}
  Let $w(\cdot|\cdot)$ define a discrete memoryless channel with inputs $\mathcal{X}$ and outputs $\mathcal{Y}$.
  Let 
  \begin{align}
    \label{eq:cap}
    R_{cap}\defeq \max_{p(x)} I(X,Y),
  \end{align}
  where the maximum is taken over probability distributions on $\mathcal{X}$, and let $\gamma>0$.
  Then, for all $\delta\in(0,1)$, there exists $n_0 = O_\delta(\frac{1}{\gamma^2})$ such that, for all $n\ge n_0$ there do not exist codes of rate $\frac{R_{cap}+\gamma}{\log|\mathcal{X}|}$ decodable with failure probability at most $\delta$ under the channel $w(\cdot|\cdot)$.\footnote{The quantity $R$ is often referred to as the \emph{capacity} of the channel}\footnote{Typically the normalizing term $\frac{1}{\log|\mathcal{X}|}$ is not present when stating Shannon's capacity theorem. This is because the ``rate'' used in Shannon capacity is often defined as $\frac{\log|C|}{n}$, whereas the rate for us is defined as $\frac{\log|C|}{n\log|\mathcal{X}|}$.}
\end{theorem}

A classic result known as Fano's inequality can be used to lower bound the mutual information $I(X,Y)$ in \eqref{eq:cap} with a quantity involving the probability of error.
The following result by Tebbe and Dwyer \cite{TebbeD68} helps bound the mutual information $I(X,Y)$ in the other direction, and is useful in our proof.
\begin{lemma}[\cite{TebbeD68}]
  \label{lem:td68}
  Let $\delta\in(0,1)$.
  Suppose we are given a probability distribution $\mathcal{D}$ over $\mathcal{X}\times \mathcal{Y}$ such that, for all maps $f:\mathcal{Y}\to \mathcal{X}$, we have $\Pr_{X,Y}[f(Y)\neq X]\ge \delta$.
  Then $H(X|Y)\ge \frac{\delta}{2}$. 
\end{lemma}

\subsection{Random to coded lower bound}
Let $\mathcal{X}_m\defeq \{0,1\}^m$ and $\mathcal{Y}_{m,T} \defeq (\{0,1\}^{\le m})^T$.
For all $q$, $m$ and $T$, there is a natural channel with inputs $\mathcal{X}_m$ and outputs $\mathcal{Y}_{m,T}$. We induce a joint probability distribution on $\mathcal{X}_m, \mathcal{Y}_{m, T}$ as follows. Let $\lambda$ be a probability density function on $\mathcal X_m$. Let $X_{\lambda} \sim \mathcal X_m$ be the distribution where $x$ is sampled with probability $\lambda(x)$. We let $Y_{\lambda}$ be the output of $T$ independent traces of the sampled $x \sim X_\lambda$ across the $\BDC_q.$ 

Note that since $H(X_\lambda) \le m$, for any distribution $X_\lambda \sim \mathcal X_m$, we have that $I(X_\lambda, Y_\lambda) \le m$. We show in Lemma~\ref{lem:lb-2} that if $T \approx T\rand_{q,0}(m)$, then this upper bound can be improved by a significant amount. This upper bound is subsequently used in Theorem~\ref{thm:lb-bin} to show a limitation of the capacity of coded trace reconstruction.

\begin{lemma}
  \label{lem:lb-2}
  Let $\beta\ge 1$.
  Suppose $T=T\rand_{q,0}(m)-1$ for $m\ge 32$.
  For all probability distributions $X_\lambda$ on $\mathcal{X}_m$, if $Y_\lambda\in \mathcal{Y}_{m,T}$ is distributed as $T$ independent traces of $X_\lambda$, then
  \begin{align}
    I(X_\lambda,Y_\lambda)\le m - \frac{1}{32 m\log m}.
  \end{align}
\end{lemma}

\begin{proof}
  Let $\mathcal{X}'$ be the elements of $\mathcal{X}$ with $\lambda (x)\ge \frac{1}{(m\log m)2^m}$.
  We consider two cases.

  \textbf{Case 1: $|\mathcal{X}'|\le 2^{m-1/3}$.}
  We have
  \begin{align}
    I(X_\lambda,Y_\lambda)
    \ &\le \  H(X_\lambda) \nonumber\\
    \ &= \  \sum_{x\in \mathcal{X}'}^{} \lambda(x)\log\frac{1}{\lambda(x)} + \sum_{x\notin \mathcal{X}'}^{} \lambda(x)\log\frac{1}{\lambda(x)} \nonumber\\
    \ &\le \  \log|\mathcal{X}'| + \sum_{x\notin \mathcal{X}'}^{} \lambda(x)\log\frac{1}{\lambda(x)} \label{eq:step2}\\
    \ &\le \ \log|\mathcal{X}'| + \sum_{x\notin \mathcal{X}'}^{} \frac{1}{(m\log m)2^m}\cdot \log ((m\log m)2^m) \label{eq:step3}\\
    \ &\le \ \log|\mathcal{X}'| + 2^m\cdot \frac{1}{m(\log m)2^m}\log(m(\log m)2^m) \nonumber\\
    \ &= \ m - \frac{1}{3} + \frac{m + \log m+\log\log m}{m\log m}
    \ < \ m - \frac{1}{3}+\frac{1}{4}
    \ < \ m - \frac{1}{32 m\log m}. \label{eq:last-step}
  \end{align}
  In \eqref{eq:step2} we used that $\sum_{x\in\mathcal{X}}^{} \lambda(x)\le 1$ and that $z\log\frac{1}{z}$ is concave.
  In \eqref{eq:step3} we used that $z\log\frac{1}{z}$ is increasing for $z<1/3$.
  In \eqref{eq:last-step} we used that $m$ is sufficiently large. 

  \textbf{Case 2: $|\mathcal{X}'|\ge 2^{m-1/3}$.}

  For this case, a similar argument appears in \cite{Holden:2018tx} (Proposition 4.1).
  Let $\sigma(x)$ be the uniform distribution on the elements of $\mathcal X$. Let $\mu(x)$ be the uniform distribution on the elements of $\mathcal{X}'$. Consider any trace reconstruction algorithm $f:\mathcal{Y}_{m,T}\to\mathcal{X}_m$. Note that 
  \begin{align}
    \Pr[f(Y_\sigma) \neq X_\sigma] &\le \frac{|\mathcal X \setminus \mathcal X'|}{|\mathcal X|} + \frac{|\mathcal X'|}{|\mathcal X|} \Pr[f(Y_\mu) \neq X_\mu].\nonumber
  \end{align}
  By definition, $T\defeq T\rand_{q,0}(m) - 1$ and $|\mathcal X'| \ge 2^{-1/3}|\mathcal X|$, so
  \begin{align}
  \Pr[f(Y_\mu) \neq X_\mu] &\ge 2^{-1/3}\Pr[f(Y_\sigma) \neq X_\sigma] - (1-2^{-1/3})
  \ge \frac{2^{-1/3}}{3} - 1 + 2^{-1/3} > \frac{1}{8}.\label{eq:bound-fix}
  \end{align}
  Let $\nu(x)$ be the probability distribution on $\mathcal{X}$ given by
  \begin{align}
    \nu(x) \ &\defeq \  \frac{\lambda(x) - \frac{1}{2m\log m}\mu(x)}{1-\frac{1}{2m\log m}}.
  \end{align}
  We have $|\mathcal{X}'|\ge \frac{1}{2}|\mathcal{X}|$, so $\mu$ assigns probability at most $\frac{2}{2^m}$ to each element of $\mathcal{X}'$.
  Since $\lambda$ assigns probability at least $\frac{1}{(m\log m)2^m}$ to each element of $\mathcal{X}'$, $\nu(x)\ge 0$ for all $x$.
  Furthermore, it is easy to check that $\sum_{x\in\mathcal{X}}^{} \nu(x)=1$, so $\nu(x)$ is a legitimate probability distribution.
  We can sample from $\lambda$ as follows: with probability $\frac{1}{2m\log m}$ sample from $\mu$, otherwise, sample from $\nu$. Thus, for any recovery algorithm $f:\mathcal{Y}_{m,T}\to\mathcal{X}_m$.
  \begin{align}
    \Pr[f(Y_\lambda)\neq X_\lambda] &= \frac{1}{2m\log m}\Pr[f(Y_\mu)\neq X_\mu] + \left(1 - \frac{1}{2m\log m}\right)\Pr[f(Y_\nu) \neq X_\nu] \nonumber\\
    &\ge \frac{1}{2m\log m}\cdot \Pr[f(Y_\mu)\neq X_\mu] \nonumber\\
    &\ge \frac{1}{2m\log m}\cdot \frac{1}{8} 
    = \frac{1}{16m\log m}.
  \end{align}
  The last inequality is uses (\ref{eq:bound-fix}). Thus, $H(X_{\lambda}|Y_{\lambda}) \ge \frac{1}{32 m\log m}$ by Lemma~\ref{lem:td68}.
  We thus may bound
  \begin{align}
    I(X_\lambda,Y_\lambda)
    = H(X_\lambda) - H(X_\lambda|Y_\lambda) 
    \le \log|\mathcal{X}| - H(X_\lambda|Y_\lambda)
    \le m - \frac{1}{32m\log m}.
  \end{align}
  This covers all cases, completing the proof.
\end{proof}

\begin{proof}[Proof of Theorem~\ref{thm:lb-bin}]
  Recall $m=\floor{\sqrt{\frac{1/\varepsilon}{128\log(1/\varepsilon)}}}$ and $T = T_{q,0}\rand(m) - 1$.
  Let $n_0'$ be the constant given by Theorem~\ref{thm:wolfowitz} with the parameter $\gamma\defeq \varepsilon m$.
  Let $n_0 \defeq m\cdot n_0' \le O(\frac{1}{\varepsilon^2})$.

  We first prove that codes of rate $1-2\varepsilon$ are not $(T,q,\delta)$ trace reconstructible when $n$ is any sufficiently large multiple of $m$.
  Let $C$ be a code that is $(T,q,\delta)$ trace reconstructible when $n\ge n_0$ is a multiple of $m$.
  We show $C$ must have rate less than $1-2\varepsilon$.
  Let $n_{out} \defeq \frac{n}{m}$.
  For each $i\in[n_{out}]$, given a codeword $c=(c_1,\dots,c_n)\in C$, let $X_i$ denote the string
  \begin{align}
    X_i \defeq c_{(i-1)m+1},c_{(i-1)m+2},\dots,c_{im}.
  \end{align}
  Let $Y_i\in\mathcal{Y}_{m,T}$ be a tuple of $T$ of strings distributed as independent traces of $X_i$ under the $\BDC_q$.
  By assumption of our code, it is possible to recover $c$ from $Y_1,\dots,Y_{n_{out}}$ with failure probability at most $\delta$: take the trace-wise concatenation of $Y_1,\dots,Y_{n_{out}}$ and use the trace reconstruction algorithm that is assumed.
  Hence, the code $C$, when interpreted as a code in $\mathcal{X}^{n_{out}}$, achieves failure probability $\delta$ on the memoryless channel $w(\cdot|\cdot)$ with inputs $\mathcal{X}_{m}$ and outputs $\mathcal{Y}_{m,T}$  where $Y$ is distributed as $T$ independent traces of $X$. 
  By Lemma~\ref{lem:lb-2}, we have
  \begin{align}
    \max_{\lambda\text{ on }\mathcal{X}_{m}}I(X_\lambda,Y_\lambda) \le m\left(1 - \frac{1}{32m^2\log m}\right) \le m(1-4\eps),
  \end{align}
  since $\eps$ sufficiently small.
  By Theorem~\ref{thm:wolfowitz}, since $\gamma\defeq \varepsilon m$ and $n_{out}\ge n_0'$, our code $C$, when interpreted as a code in $\mathcal{X}^{n_{out}}$, must have rate less than 
    \begin{align}
    \frac{1}{\log|\mathcal{X}|}\left(\max_{\lambda\text{ on }\mathcal{X}_{m}}I(X_\lambda,Y_\lambda) + \gamma\right) < 1 - 2\eps,
    \end{align}
    as desired.

Now suppose $n$ is not a multiple of $m$. 
Then, suppose for contradiction that $C\subset \{0,1\}^n$ is a code of length $n$ and rate $1-\varepsilon$ that is $(T,q,\delta)$ trace reconstructible.
By a simple counting argument, there exists a code $C'\subset \{0,1\}^{n'}$ of rate $1-\varepsilon-\frac{\varepsilon n'}{n-n'} > 1-2\varepsilon$ and a string $w$ such that $c'||w\in C$ for all $c'\in C'$.
Furthermore, recovering all codewords of $C$ requires recovering all codewords of the form $c'||w$ for $c'\in C'$.
The failure probability of recovering $c'$ from $T$ traces of $c'||w$ is at least the failure probability of recovering $c'$ from $T$ traces of $c'$, which, as we showed, is more than $\delta$, a contradiction.
\end{proof}

\section{Conclusion and Open Problems}

In this paper, we considered the coded trace reconstruction problem. We obtain lower and upper bounds on the problem which show that the average-case trace reconstruction problem is essentially equivalent to the coded trace reconstruction problem. Even with this contribution, there are still many questions left unanswered.

\begin{enumerate}
    \item  The most fundamental open question in this space is closing the exponential gaps for the worst-case trace reconstruction and average-case trace-reconstruction. For worst-case trace reconstruction, the optimal number of traces is between $\tilde\Omega(n^{3/2})$ and $\exp(O(n^{1/3}))$ (or $\exp(O(n^{1/5}))$ for $q\le 1/2$), and for average-case trace reconstruction, the optimal number of traces is between $\tilde\Omega(\log^{5/2}n)$ and $\exp(O(\log^{1/3}n))$. %

    \item One way to generalize the coded trace reconstruction model considered in this paper to consider a more general synchronization channel, such as with insertions and deletions. For example, such a model could insert $k$ random bits between $x_i$ and $x_{i+1}$ with probability $(1-q)q^{k}$ and then apply i.i.d. deletions with probability $q$. See the recent survey by Cheraghchi and Ribeiro~\cite{cheraghchi2019overview} for an overview of various models for random insertions, deletions, substitutions and replications. The authors suspect that similar primitives to those used in this paper could be useful in these more general settings.

    \item Another combinatorial variant of this question is \emph{necklace reconstruction}. This question is similar to ordinary trace reconstruction, except a random cyclic shift is also applied to each trace, and the original string needs to be recovered up to an arbitrary cyclic shift. Many protocols for the traditional trace reconstruction problem exploit that the initial prefix of the trace can be easily determined by looking at the prefixes of the traces. For necklace reconstruction, this strategy would no longer work (due to the random shift), so new techniques need to be developed. Even beating $O((1-q)^{-n})$ traces, the probability of receiving the whole necklace as a trace, seems nontrivial. A recent paper~\cite{narayanan2020circular} studies this problem.
    
    \item A challenging question in the context of coded trace reconstruction is formulating other interesting models beyond i.i.d. deletions. Adversarial deletions is not an interesting model because the adversary could delete the same bits on each trace, reducing the problem to the deletion code problem. One possibility of such a model would be adversarial deletions subject to some global constraints--such as the distribution of deletions being approximately $k$-wise independent.
    
    \item Another challenge is coming up with deletion models and codes that more accurately correspond to practical use cases and string lengths. Trace reconstruction as used in DNA computing often considers string of approximately length $100$ (e.g., \cite{organick2018random}). Constructing such codes may require different techniques than those used in this paper.
    
    \item We do not know if Theorem~\ref{thm:large-ub} achieves the smallest alphabet size for $O(\log_{1/q}\frac{1}{\varepsilon})$ traces. It would be interesting to determine the trade-off between alphabet size and number of traces.
\end{enumerate}

\section{Acknowledgements}

We thank Mary Wootters for sponsoring the Coding Theory Reading Group at Stanford where this project was started.
We thank Nina Holden for helpful discussions on the error probabilities in the paper \cite{HPP18}.
We thank Venkatesan Guruswami for help discussions on high rate error correcting codes and suggesting the Justesen code construction in Proposition~\ref{prop:justesen}.
We thank Jo\~ao Ribeiro for helpful discussions about the work \cite{Cheraghchi:2019vd} and feedback on an earlier draft of the paper.
We thank Wesley Pegden for suggesting one of the open problems.
We thank Venkatesan Guruswami, Mary Wootters, Aviad Rubinstein, Moses Charikar, and Sivakanth Gopi for helpful discussions, encouragement, and feedback on an earlier draft of the paper.

\bibliography{trace_biblio}
\bibliographystyle{alpha}

\appendix

\section{Omitted Details}\label{app:missing}
\subsection{Proof of Lemma~\ref{lem:short}}
\label{app:short}
The following lemma shows that a significant fraction of strings of length $n$ are $m$-protected for $m\ge \Omega(\log n)$.
\begin{lemma}
  \label{lem:pad}
  Let $m\ge 10^3$ be an integer and $n\in[3m, 2^{m/150}]$, the number of $m$-protected codewords is at least $2^{n-2m-3}$.
\end{lemma}
\begin{proof}
  There are $2^{n-2m-2}$ strings of the form  $s= 0^{m}s^\circ1^m$.
  Choose one such string at random, so that $s^\circ$ is a uniformly random string in $1||\{0,1\}^{n-2m-2}||0$.
  For a substring of length $m'$, the probability it has at least $3/4$ or at most $1/4$ fraction of 1s is, by the Chernoff bound \eqref{eq:cher-1}, at most $2\cdot 2^{-m'/16}$.
  Since $m'\ge m/4$, this is at most $2\cdot 2^{-m/64}$.
  Since there are at most $n^2$ substrings of $s^\circ$ of length at least $m/4$, by the union bound, the resulting string is \emph{not} $m$-protected with probability at most
  $$2n^2\cdot 2^{-m/64} \leq 2\cdot 2^{2m/150-m/64} = 2\cdot 2^{-m/300} < \frac{1}{2}.$$
  Hence, at least half of all strings of the form $0^ms^\circ1^m$ are $m$-protected, as desired.
\end{proof}

With the protected strings from Lemma~\ref{lem:pad} and the codes for trace reconstruction from Lemma~\ref{lem:avg-tr}, we can prove Lemma~\ref{lem:short}. Intuitively, both the codes with protected codewords and codes which are efficiently trace reconstructible are both very large, so we can find the desired code in their intersection.

\begin{proof}[Proof of Lemma~\ref{lem:short}]
  By Lemma~\ref{lem:avg-tr} with parameter $\beta'=3\beta$, there exists a code $C_1$ with $|C_1|\ge (1-n^{-3\beta})2^n = 2^n - 2^{n-3m}$ that is $(T,q,n^{-3\beta})$ trace reconstructible.
  
  Let $C_2$ be the set of length $n$ strings that are $m$-protected.
  Assume $\varepsilon$ is sufficiently small so that $n \ge 6\frac{1}{\varepsilon}m > 3m$.
  Note also that by our choice of $n$,
  $$2^{m/150} = 2^{\floor{\beta \log n}/150} \geq 2^{\log n} = n.$$
  Then, by Lemma~\ref{lem:pad}, we have $|C_2|\ge 2^{n-2m-3}$.
  
  Let $C = C_1\cap C_2$.
  We have $|C| = |C_1\cap C_2| \ge 2^{n-2m-3} - 2^{n-3m} > 2^{n-3m}$, so $C$ has rate at least $1 - \frac{3m}{n} > 1 - \frac{\varepsilon}{2}$.
  Furthermore, since $C_1$ is $(T,q,n^{-3\beta})$ trace reconstructible, $C$ is as well.
\end{proof}

\subsection{Proof of Lemma~\ref{lem:bdc}}
\label{app:bdc}
In this section, we show how to construct codes for the binary deletion channel of length $O(\log \frac{1}{\delta})$ and failure probability at most $\delta$.
\begin{proof}[Proof of Lemma~\ref{lem:bdc}]
  \textbf{Encoding.}
  Map every element $\sigma\in[2^K]$ to a string $\tilde c_\sigma\in\{0,1\}^{3K}$ that starts with a 0, ends with a 1, has $K$ runs are length 1, and has $K$ runs are length 2.
  There are $\binom{2K}{K} \ge 2^K$ such strings as each string is uniquely determined by its sequence of run lengths, so each $\sigma$ can be assigned to a distinct string.
  Let $c_\sigma$ be $\tilde c_{\sigma}$ with every symbol duplicated $m$ times.

  \textbf{Decoding.}
  To decode a received word $s$ under the BDC$_q$, we first recover $\tilde c_\sigma$, and then recover $\sigma$.
  To recover $\tilde c_\sigma$, suppose $s$ is of the form $0^{k_1'}1^{\ell_1'}\cdots 0^{k_K'}1^{\ell_K'}$ where $k_i',\ell_i'\ge 1$ for all $i$.
  If $s$ is not of this form, return an arbitrary symbol in $[2^K]$ (give up).
  For each $i=1,\dots,K$, if $k_i'\ge 1.4(1-q)m$, let $x_i'=2$, and otherwise let $x_i'=1$.
  Similarly, if $\ell_i'\ge 1.4(1-q)m$, let $y_i'=2$, and otherwise let $y_i'=1$.
  The decoding returns the symbol $\sigma'$ such that
  \begin{align}
    \tilde c_{\sigma'} = 0^{x_1'}1^{y_1'}\cdots 0^{x_K'}1^{y_K'}.
  \end{align}

  \textbf{Analysis.}
  The decoding is clearly linear time.
  To prove correctness, suppose our input symbol $\sigma$ satisfies $c_\sigma = 0^{x_1}1^{y_1}\cdots 0^{x_K}1^{y_K}$, where $x_i,y_i\in\{1,2\}$ for all $i$.
  Let $k_1,\ell_1,\dots,k_K,\ell_K$ denote the number of bits not deleted in the corresponding runs $0^{x_1},1^{y_1},\dots,1^{y_K}$.
  We bound the probability each of the following happen.
  \begin{enumerate}
  \item There exists some $i$ such that $k_i=0$ ($\ell_i=0$)
  \item There exists some $i$ with $x_i=1$ ($y_i=1$) such that $k_i\ge 1.4(1-q)m$ ($\ell_i\ge 1.4(1-q)m$).
  \item There exists some $i$ with $x_i=2$ ($y_i=2$) such that $k_i< 1.4(1-q)m$ ($\ell_i< 1.4(1-q)m$).
  \end{enumerate}
  If $x_i=1$, then $k_i$ is distributed as the binomial distribution $B(m, 1-q)$.
  If $x_i=2$, then $k_i$ is distributed as the binomial distribution $B(2m, 1-q)$.
  In either case, we have
  \begin{align}
    \Pr[k_i=0] =\Pr[\ell_i=0]\le q^m < e^{-(1-q)m} < 2^{-(1-q)m/20}
  \end{align}
  By the Chernoff bound \eqref{eq:cher-2}, for $i$ such that $x_i=1$, we have
  \begin{align}
    \Pr[x_i\neq x_i'] 
    \ &= \ \Pr[k_i\ge 1.4(1-q)m]
    \ \le \ e^{-\frac{0.4^2}{2+0.4}(1-q)m}
    \ < \ 2^{-(1-q)m/20}
  \end{align}
  On the other hand, for $i$ such that $x_i=2$, we have, by the Chernoff bound \eqref{eq:cher-1},
  \begin{align}
    \Pr[x_i\neq x_i'] 
    \ &= \ \Pr[k_i < 1.4(1-q)m]
    \ \le \ e^{-\frac{0.3^2}{2}(1-q)m}.
    \ < \ 2^{-(1-q)m/20}
  \end{align}
  The same probabilities hold for $y_i$'s.
  Hence the probability any of events 1, 2, or 3 happen is at most $6K\cdot 2^{-(1-q)m/20}$, as desired.
  However, if event 1 does not happen then the decoding guarantees that $k_i'=k_i$ and $\ell_i'=\ell_i$ for all $i$.
  If additionally, events 2 and 3 do not happen, the decoding guarantees that $x_i'=x_i$ and $y_i'=y_i$, and hence $\sigma'=\sigma$.
  Thus, the decoding fails with probability at most $6K\cdot 2^{-(1-q)m/40}$, as desired.
\end{proof}

\subsection{High rate error correcting codes}
\label{app:ecc}
In this section, we show how Proposition~\ref{prop:gi} follows from the construction of Guruswami and Indyk \cite{GI05}.
Guruswami and Indyk prove the following.
\begin{theorem}[Theorem 5 of \cite{GI05}]
  \label{thm:gi-2}
For every $\varepsilon>0$ and any $R\in(0,1)$, there exists a family of binary codes of rate $R$ encodable in linear time and decodable in linear time from up to a fraction $\delta$ of substitution errors, where
\begin{align}
  \delta\ge \max_{R < r<  1}\frac{(1-r-\varepsilon)H^{-1}(1-R/r))}{2}.
\end{align}
\end{theorem}
By setting $R = 1-\varepsilon'$ and $\varepsilon = \frac{\varepsilon'}{10}$, and taking $r=1-\frac{\varepsilon'}{2}$, we have
\begin{align}
  \delta
  \ge \frac{2\varepsilon'}{5}\cdot H^{-1}\left(\frac{\varepsilon'}{2(1-\varepsilon'/2)}\right)
  \ge \frac{2\varepsilon'}{20}\cdot \left(\frac{\varepsilon'}{2(1-\varepsilon'/2)}\right)^2
  \ge \frac{2\varepsilon'}{20}\cdot \left(\frac{\varepsilon'}{2}\right)^2
  \ge \frac{(\varepsilon')^3}{40}. 
\end{align}
Here we used that $H^{-1}(x) \ge \frac{x^2}{4}$ for all $x\in(0,1)$.

Further for every $\Sigma$ whose size is $2^\ell$ a power of 2, every family binary codes of rate $R$ and tolerating a $\delta$ fraction of worst-case substitution errors can be made into a family of codes over $\Sigma$ with the same asymptotic rate and error tolerance: pad each codeword so that its length is a multiple of $\ell$ (this has a negligible effect on the asymptotic rate and error tolerance), then map each length $\ell$ string $b_1,\dots,b_\ell$ to a unique element of $\Sigma$. For a codeword $c=(c_1,\dots,c_n)\in\{0,1\}^n$, create a codeword over $\Sigma^{n/\ell}$ whose $i$th symbol is the image of $c_{(i-1)\ell+1}, c_{(i-1)\ell+2},\dots,c_{i\ell}$ under this mapping.
Then to correct a string in $\Sigma^{n/\ell}$,  interpret it as a binary string of length $n$: $\delta$ fraction of substitution errors in a codeword in $\Sigma^{n/\ell}$ yields at most a $\delta$ fraction of worst-case substitution errors over the underlying binary codeword, which can be corrected by assumption.

We now prove Proposition~\ref{prop:justesen}.
\cite{Justesen72}
\begin{lemma}
  For all positive integers $s\le m$, there exists a linear code $C:\mathbb{F}_{2^m}\to \mathbb{F}_2^{m+s}$ of dimension $m$ and length $m+s$ tolerating $\half\floor{(m+s)\cdot H^{-1}(\frac{s}{m+s})}$ errors.
  Furthermore, such a code can be found in time $\tilde O(2^{2m})$.
\label{lem:just-1}
\end{lemma}
\begin{proof}
  Since $\mathbb{F}_{2^m}$ is a $\mathbb{F}_2$ vector space, there exists a linear bijection $\sigma:\mathbb{F}_{2^m}\to \mathbb{F}_2^m$.
  Let $\sigma':\mathbb{F}_{2^m}\to\mathbb{F}_2^s$ be given by taking the first $s$ bits of $\sigma(x)$.
  Let $e \defeq \floor{(m+s)\cdot H^{-1}(\frac{s}{m+s})}$.

  For $\alpha\in\mathbb{F}_{2^m}$, let $C_\alpha$ be the code given by the encoding $\Enc_\alpha:\mathbb{F}_{2}^m\to \mathbb{F}_{2}^{m+s}$ with
  \begin{align}
    x\mapsto (x, \sigma'(\alpha\cdot\sigma^{-1}(x))).
  \end{align}
  Since multiplication by $\alpha$ is bijective and $\mathbb{F}_2$ linear, and $\sigma$ and $\sigma'$ are linear, all such codes $C_\alpha$ are linear.
  For any $x\in\mathbb{F}_2^m$, for a random $\alpha\in\mathbb{F}_{2^m}$, we have $\alpha\sigma^{-1}(x)$ is uniform on $\mathbb{F}_2^s$, so $\sigma'(\alpha\sigma^{-1}(x))$ is uniform on $\mathbb{F}_2^s$.
  Thus, each element of $\mathbb{F}_2^{m+s}$ appears exactly $2^{m-s}$ times in $\{C_\alpha:\alpha\in\mathbb{F}_q\}$ 
  Let $X_{bad}$ denote the set of nonzero element of $\mathbb{F}_2^{m+s}$ with Hamming weight at most $e$.
  This set has size at most $\sum_{i=1}^{e} \binom{m+s}{i}< 2^{(m+s)H(\frac{e}{m+s})}$.
  Thus, for a uniformly random $\alpha\in\mathbb{F}_2^m$, the probability that there exists a nonzero element of $X_{bad}$ in $C_\alpha$
  \begin{align}
    \frac{|X_{bad}|\cdot 2^{m-s}}{2^m} < \frac{2^{(m+s)H(\frac{e}{m+s})}}{2^s} \le 1
  \end{align}
  Hence, there exists some $\alpha$ such that $C_\alpha$ has no elements in $X_{bad}$.
  We can find such an $\alpha$ by brute force in time $\tilde O(2^{2m})$: each $\alpha$ takes time $\tilde O(2^m)$ to compute all codewords and check their hamming weight, and there are $2^m$ such $\alpha$.
  In $C_\alpha$, any two codewords have Hamming distance at least $e$, so it tolerates up to $\frac{e}{2}$ errors.
\end{proof}

\begin{proof}[Proof or Proposition~\ref{prop:justesen}]
  Let $m$ be the smallest integer larger than $\frac{12}{\varepsilon}$ such that $m\cdot 2^m\ge n$.
  Let $s$ be the largest integer such that $\frac{m}{m+s}\ge 1 - \frac{\varepsilon}{3}$, so that $\frac{s}{m+s} \ge \frac{\varepsilon}{4}$.
  
  By Lemma~\ref{lem:just-1}, there exists a code $C_{in}:\mathbb{F}_{2^m}\to \mathbb{F}_2^{m+s}$ of dimension $m$ and length $m+s$ with minimum distance $\floor{(m+s)H^{-1}(\frac{s}{m+s})}$. 
  Let $C_{out}:$ be a Reed Solomon code over $\mathbb{F}_{2^m}$ of length $n'\defeq \floor{n/(m+s)}$ and dimension $k'=\ceil{n'(1-\frac{\varepsilon}{3})}$.
  Let $C\subset \{0,1\}^n$ be the concatenation of $C_{in}$ and $C_{out}$ with $n-n'm$ 0s padded on the end.
  The code $C_{in}$ has rate $\frac{m}{m+s} > 1 - \frac{\varepsilon}{3}$.
  The code $C_{out}$ has rate at least $1 - \frac{\varepsilon}{3}$.
  The padding of $0$s multiplies the rate by $\frac{n'(m+s)}{n} \ge 1 - \frac{m+s}{n} > 1 - \frac{\varepsilon}{3}$.
  Thus, the total rate is at least $(1-\frac{\varepsilon}{3})^3 > 1-\varepsilon$.

  To decode a received word $c_1,\dots,c_n$, we first run the inner decoding to obtain symbols $\alpha_i\in\mathbb{F}_{2^m}$ for $i=1,\dots,n'$, where $\alpha_i$ is the decoding of $c_{(i-1)(m+s)+1},\dots,c_{i(m+s)}$ under $C_{in}$.
  Then, we run the outer decoding on $\alpha_1,\dots,\alpha_{n'}$ to obtain the message.
  The inner decoding can be computed by brute force in time $O(m2^{m})<O_\varepsilon(n)$.
  The outer decoding can be computed in time $O(n^2)$ using the Berlekamp-Massey algorithm.
  Thus, the total decoding run time is $O_\varepsilon(n^2)$.
  The encoding takes time $O_\varepsilon(n^2)$, and construction takes time $\tilde 
  O(2^{2m}) = O_\varepsilon(n^2)$ because we need to construct the inner code.
  
  The outer code tolerates $n' - k' > \frac{\varepsilon n'}{4}$ errors.
  The inner code tolerates up to $\frac{1}{2}\floor{(m+s)H^{-1}(\frac{s}{m+s})}$ errors, and thus every incorrect $\alpha_i$ accounts for at least $\half (m+s)H^{-1}(\frac{s}{m+s}) > \half(m+s)H^{-1}(\frac{\varepsilon}{4})$ errors.
  Thus, for the outer decoding to fail, we need at least $\frac{\varepsilon n'}{4} \cdot \frac{1}{2}(m+s)\cdot H^{-1}(\frac{\varepsilon}{4}) > \frac{\varepsilon^2 n }{500\log\frac{1}{\varepsilon}}$ errors.
  Here, we used that $(m+s)n' > 0.9n$ and that $H^{-1}(x) > \frac{x}{2\log(6/x)}$, so $H^{-1}(\frac{\varepsilon}{4}) > \frac{\varepsilon}{48\log(1/\varepsilon)}$.
\end{proof}

\end{document}